\documentclass[10pt,a4paper]{amsart}
\usepackage[utf8]{inputenc}
\usepackage{amsmath}
\usepackage{amsfonts}
\usepackage{amssymb}
\usepackage{amsthm}
\usepackage{graphicx}
\usepackage{dsfont}
\usepackage{stmaryrd}
\usepackage{xcolor}
\usepackage{tikz}
\usepackage{subcaption}
\usetikzlibrary{patterns}
\usetikzlibrary{backgrounds,scopes}
\usetikzlibrary{calc}
\usetikzlibrary{intersections}

\usepackage{hyperref}

\usepackage[all]{xy}

\def\res{\mathop{{\rm Res}}}

\newcommand{\beq}{\begin{equation}}
\newcommand{\eeq}{\end{equation}}
\newcommand{\lb}{\left (}
\newcommand{\rb}{\right )}

\newcommand{\dif}[2]{\dfrac{d #1 }{d #2 } }
\newcommand{\pdif}[2]{\dfrac{\partial #1 }{\partial #2 } }
\newcommand{\la}{\lambda}

\newcommand{\Bc}{\check{B}}
\newcommand{\br}[1]{\lb #1 \rb }

\newcommand{\pdub}{p_{\rm \hspace{.1mm}Dub}}

\makeatletter
\newcommand{\vast}{\bBigg@{4}}
\newcommand{\Vast}{\bBigg@{5}}
\makeatother

\newtheorem{theorem}{Theorem}[section]
\newtheorem{proposition}[theorem]{Proposition}
\newtheorem{corollary}[theorem]{Corollary}
\newtheorem{lemma}[theorem]{Lemma}
\newtheorem{example}[theorem]{Example}

\theoremstyle{definition}
\newtheorem{definition}[theorem]{Definition}

\theoremstyle{remark}
\newtheorem{remark}[theorem]{Remark}

\newcommand{\C}{\mathbb{C}}
\newcommand{\N}{\mathbb{N}}
\newcommand{\R}{\mathbb{R}}

\newcommand{\oM}{\overline{\mathcal{M}}}
\newcommand{\diag}{\mathrm{diag}}
\newcommand{\Id}{\mathbf{I}}

\newcommand{\ep}{\epsilon}

\newcommand{\un}{1\!\!1}

\def\picsize{1.8cm}
\def\rectsize{1.1}

\newcommand\encadremath[1]{\vbox{\hrule\hbox{\vrule\kern8pt
\vbox{\kern8pt \hbox{$\displaystyle #1$}\kern8pt}
\kern8pt\vrule}\hrule}}
\def\enca#1{\vbox{\hrule\hbox{
\vrule\kern8pt\vbox{\kern8pt \hbox{$\displaystyle #1$}
\kern8pt} \kern8pt\vrule}\hrule}}

\newcommand\figureframex[3]{
\begin{figure}[bth]
\hrule\hbox{\vrule\kern8pt
\vbox{\kern8pt \vbox{
\begin{center}
{\mbox{\epsfxsize=#1.truecm\epsfbox{#2}}}
\end{center}
\caption{#3}
}\kern8pt}
\kern8pt\vrule}\hrule
\end{figure}
}
\newcommand\figureframey[3]{
\begin{figure}[bth]
\hrule\hbox{\vrule\kern8pt
\vbox{\kern8pt \vbox{
\begin{center}
{\mbox{\epsfysize=#1.truecm\epsfbox{#2}}}
\end{center}
\caption{#3}
}\kern8pt}
\kern8pt\vrule}\hrule
\end{figure}
}

\newcommand{\bea}{\begin{eqnarray}}
\newcommand{\eea}{\end{eqnarray}}

\renewcommand{\and}{{\qquad {\rm and} \qquad}}

\newcommand{\Res}{\mathop{\,\rm Res\,}}

\newcommand{\om}{\omega}

\newcommand{\Pint}{{\int\kern -1.em -\kern-.25em}}

\newcommand{\cal}{\mathcal}

\makeatletter
\@addtoreset{equation}{section}
\makeatother
\newcommand{\brem}{\begin{remark}\rm\small}
\newcommand{\er}{\end{remark}}
\newcommand{\bt}{\begin{theorem}}
\newcommand{\et}{\end{theorem}}
\newcommand{\bd}{\begin{definition}}
\newcommand{\ed}{\end{definition}}
\newcommand{\bp}{\begin{proposition}}
\renewcommand{\ep}{\end{proposition}}
\newcommand{\bl}{\begin{lemma}}
\newcommand{\el}{\end{lemma}}
\newcommand{\bc}{\begin{corollary}}
\newcommand{\ec}{\end{corollary}}
\newcommand{\beaq}{\begin{eqnarray}}
\newcommand{\eeaq}{\end{eqnarray}}

\newcommand{\refl}{\mathcal{R}}

\begin{document}
	
\title{Dubrovin's superpotential as a global spectral curve}

\author[P.~Dunin-Barkowski]{P.~Dunin-Barkowski}
\address{P.~D.-B.:
Faculty of Mathematics, National Research University Higher School of Economics, Usacheva 6, 119048 Moscow, Russia}
\email{ptdbar@hse.ru}

\author[P.~Norbury]{P.~Norbury}
\address{P.~N.: Department of Mathematics and Statistics, University of Melbourne, Australia 3010}
\email{pnorbury@ms.unimelb.edu.au}

\author[N.~Orantin]{N.~Orantin}
\address{N.~O.:  D\'epartement de math\'ematiques, 
Ecole Polytechnique F\'ed\'erale de Lausanne,
CH-1015 Lausanne,
Switzerland}
\email{nicolas.orantin@epfl.ch}

\author[A.~Popolitov]{A.~Popolitov}
\address{A.~P.: Korteweg-de Vries Institute for Mathematics, University of Amsterdam, Postbus 94248, 1090 GE Amsterdam, The Netherlands and ITEP, Moscow, Russia}
\email{A.Popolitov@uva.nl}

\author[S.~Shadrin]{S.~Shadrin}
\address{S.~S.: Korteweg-de Vries Institute for Mathematics, University of Amsterdam, Postbus 94248, 1090 GE Amsterdam, The Netherlands}
\email{S.Shadrin@uva.nl}

\begin{abstract}
We apply the spectral curve topological recursion to Dubrovin's universal Landau-Ginzburg superpotential associated to a semi-simple point of any conformal Frobenius manifold. We show that under some conditions the expansion of the correlation differentials reproduces the cohomological field theory associated with the same point of the initial Frobenius manifold. 
\end{abstract}

\maketitle

\tableofcontents

\section{Introduction}

\subsection{Goal of the paper}
A semi-simple (conformal) Frobenius manifold is an important algebro-geometric structure, introduced by Dubrovin, that appears naturally in a circle of questions related to classical mirror symmetry.  Closely related to a semi-simple conformal Frobenius manifold is a cohomological field theory, that is, a system of cohomology classes on the moduli space of stables curves introduced by Kontsevich and Manin in order to capture the main universal properties of Gromov-Witten theory. Via Givental-Teleman theory, these two concepts (semi-simple conformal Frobenius manifolds and semi-simple homogeneous cohomological field theories) are essentially equivalent. 

The theory of Landau-Ginzburg superpotentials associates to a Riemann surface (or a family of Riemann surfaces) equipped with a meromorphic function and a meromorphic differential 1-form (or a meromorphic function whose differential is this 1-form) 
 structure that is essentially equivalent to the concept of a semi-simple Frobenius manifold, after work of Dubrovin \cite{Dub2dTFT}. It is part of a more general theory of Landau-Ginzburg models that exists in any dimension, not necessarily on a curve.

The theory of spectral curve topological recursion, initially developed for computation of the correlation differentials of matrix models, uses a very similar input: a Riemann surface (or a family of Riemann surfaces) equipped with a meromorphic function, a meromorphic differential 1-form (or a meromorphic function, whose differential is this 1-form), and a symmetric bi-differential. It produces a system of symmetric differentials on the cartesian powers of the underlying Riemann surface. Under some extra conditions these symmetric differentials can be expressed in terms of the correlators of a cohomological field theory.

To summarize, we have the following system of relations:
\begin{equation}\label{eq:table-of-correspondences}
\begin{array}{ccc}
\text{semi-simple conformal} & \leftrightarrow & \text{Landau-Ginzburg} \\
\text{Frobenius manifolds (FM)} & & \text{superpotentials (LG)}\\
\Updownarrow & & \\
\text{semi-simple homogeneous} & & \text{spectral curve} \\
\text{cohomological field theories (CohFT)} & \leftrightarrow & \text{topological recursion (TR)}
\end{array}
\end{equation}
We give precise definitions of all geometric structures involved in this diagram and explain the precise statements about their relations in Section~\ref{sec:Recollection}.  In all cases the rigorous formulation of these correspondences requires extra conditions and is not a one-to-one correspondence or an equivalence of categories. It is more like a dictionary that allows one to translate from one language to another under various extra assumptions.

The theory of Landau-Ginzburg superpotentials and spectral curve topological recursion use almost the same input data, namely a Riemann surface equipped with a meromorphic function and a meromorphic differential 1-form.  This input data is used in a completely different way in these two theories, nevertheless the natural question is whether one can add a vertical arrow so that the diagram commutes.  More explicitly,
if a Landau-Ginzburg superpotential and spectral curve topological recursion produce the same Frobenius manifold/CohFT structure on the left hand side of this diagram, do we expect that the input data for the LG model and TR to be the same? 

This paper is devoted to an affirmative answer to this question. As in the case of all other correspondences in this diagram, it is not an equivalence of categories or one-to-one correspondence, but rather a system of general statements that allows one to connect the input data of LG and TR in a large class of examples. 

\subsection{Description of the main results}

A Landau-Ginzburg superpotential determines a structure of a semi-simple conformal Frobenius manifold. However, a particular semi-simple Frobenius structure can have several different superpotentials, and it is not at all clear that it always has a superpotential. The latter problem was addressed by Dubrovin in~\cite{Dub98}, where he proposed a general construction that under some mild extra assumptions associates a superpotential to a Frobenius manifold. This construction is called \emph{Dubrovin's superpotential} in this paper. It is a family of curves $\mathcal{D}_\tau=\mathcal{D}(\tau)$ parametrized by the semi-simple points $\tau$ of the underlying Frobenius manifold, equipped with two meromorphic functions, $\lambda_\tau$ and $p_\tau$ (in fact, the differential form $dp_\tau$ would be sufficient for the definition). This construction depends on some choices, which are an important part of the LG-TR correspondence presented in this paper. 

For the spectral curve topological recursion we need a Riemann surface $\Sigma$ with two meromorphic functions, $x$ and $y$ (in fact, the differential forms $dx$ and $dy$ are sufficient for the construction), and a symmetric bi-differential $B$ on $\Sigma\times\Sigma$ with double pole with bi-residue 1 on the diagonal~\cite{EO,EO09,Ey}. 

To every point of the Frobenius manifold we associate a cohomological field theory $\alpha_\tau$, using the Givental-Teleman theory (see~\cite{GiventalMain,Teleman,Shadrin09,DSS13,PandPixZvo}). Under some conditions one can also associate a CohFT to the spectral curve topological recursion~\cite{DOSS12} (see also~\cite{MilanovAncestor}, where an approach using singularity theory is given, and~\cite{EynardIntersection}, where a more general framework for this correspondence is discussed). 

The mains results of this paper, Theorems~\ref{thm:genus-0-identification} and \ref{thm:any-genus-identification}, are devoted to proving that topological recursion applied to 
$$
\Sigma=\mathcal{D}_\tau, \quad x=\lambda_\tau, \quad y=p_\tau,
$$
and some choice of $B$ gives, under the correspondence from~\cite{DOSS12}, exactly the CohFT $\alpha_\tau$.  

One main tool in the proofs of Theorems~\ref{thm:genus-0-identification} and \ref{thm:any-genus-identification} is the result of \cite{DOSS12} which associates a CohFT to topological recursion applied to a spectral curve satisfying a set of conditions.  These conditions give a close relation between $x$, $y$ and $B$.  Our main task is to show that $\lambda_\tau$, $p_\tau$ and an appropriately chosen $B$ satisfy these relations.  Here we give a brief outline of the proof.

The identification of $x$ and $\lambda$ up to some topological properties is the starting point since the CohFT is based on a vector space formally spanned by the zeros of $dx$, respectively the zeros of $d\lambda$.  On the side of topological recursion there is one requirement that we need, namely we have to assume that there is exactly one critical point  on $\mathcal{D}_\tau$ over each critical value of $x=\lambda_\tau$\footnote{We release this constraint in section \ref{sec:general-theory}.}. This gives a restriction on the possible choices of analytic continuation in Dubrovin's superpotential.

The relation of $y$ with structure constants in the Frobenius manifold required in \cite{DOSS12} leads to an identification of $y=p_\tau$. This theorem (Theorem~\ref{thm:right-function-y}) is heavily based on the computations done by Dubrovin in~\cite{Dub98}.  Next we need to find a good choice of $B$ that will make either theorem work.  In genus zero we find that the unique possible Bergman kernel $B$ satisfies the conditions required by \cite{DOSS12} which we present in a form that can be checked (or used as a condition) for the superpotentials. This is Theorem~\ref{thm:CompatibilityTest} and its corollaries.  It allows us to conclude that topological recursion applied to the superpotential produces a CohFT and it remains to prove that this CohFT is the one associated to the Frobenius manifold defined by the superpotential.  We show that in fact it is sufficient to know that we get \textit{homogeneous} CohFT from the Bergman kernel -- then the correct CohFT $\alpha_\tau$ is reproduced automatically. This leads to a general theorem on the LG-TR correspondence in genus 0 (Theorem~\ref{thm:genus-0-identification}). This theorem is key to several important examples that we discuss in this paper as well (we mention these examples in the list of applications in Section~\ref{sec:Applications}).

In higher genera, the Bergman kernel is not canonical and we need to choose the correct one.  In order to have a suitable shape of the Laplace transform of the Bergman kernel (required for correspondence with Givental graphs), we have to use the Bergman kernel normalized on a basis of $\mathcal{A}$-cycles for some Torelli marking, using results of Eynard~\cite{EynardIntersection}. We show, using the Rauch variational formulae, that the homogeneity property is also satisfied in this case, and this allows us to make a general statement for the LG-TR correspondence in any genus (Theorem~\ref{thm:any-genus-identification}). This is a conditional statement requiring Theorem~\ref{thm:CompatibilityTest} that needs to be checked in particular examples. Still there are interesting examples, in particular, we work out an elliptic example in detail (Theorem~\ref{th:elliptic}).

Finally, we develop a theory for the case when the extra assumptions on the choice of analytic continuation of Dubrovin's superpotential are dropped. In this case we have to generalize the set-up of topological recursion in order to take into account the action of the reflection group associated with Frobenius manifold. The correspondence that we obtain in this case (Theorem~\ref{thm:general-set-up}) is parallel to the ideas of Milanov~\cite{MilanovGlobal}.

\subsection{Contributions to the theory of topological recursion}\label{sec:Applications}

In order to establish a correspondence with the Landau-Ginzburg theory and to work out several basic examples, we obtain a number of results that are of independent interest for the theory of topological recursion, and here we collect them all.

\subsubsection{Global spectral curve for the CohFT-TR correspondence}
One way to present our main result is the following. The correspondence between CohFT and TR obtained in~\cite{DOSS12} uses a local version of topological recursion, that is when the spectral curve  is just a union of disks. An important open question is whether we can glue all these open disks into a global spectral curve. This would allow one to use a variety of analytical methods developed in the theory of topological recursion that are applicable only in the case of a global curve~\cite{EO,EO09}. The main result of our paper is an affirmative answer to this question, that is, for a large class of CohFTs we can indeed claim the existence of a global spectral curve. In this form this question was also considered by Milanov for singularity theory~\cite{MilanovGlobal}.

\subsubsection{Bouchard-Eynard recursion locally}
Topological recursion requires the spectral curve to have simple critical points.
There is an extension of the theory of topological recursion for the curves with higher order critical points, due to Bouchard and Eynard~\cite{BouchardEynard}.  A fundamental question is to identify the correlation functions of their generalized recursion in the elementary case of one point of order $r+1$. Bouchard and Eynard have announced~\cite{BouchardEynardThm} a theorem that in this case the correlators are expanded in terms of the string tau-function of the $r$-Gelfand-Dickey hierarchy (or, equivalently, in terms of the intersection theory of the Witten top Chern class on the moduli space of $r$-spin structures,~\cite{WittenConj,FSZ}).  

An application of the main theorem of this paper, i.e. where topological recursion applied to Dubrovin's construction of a superpotential produces the same CohFT is the case of the $A_n$ singularity. Careful analysis of this example in its limit at the zero point implies immediately the theorem of Bouchard and Eynard. 

\subsubsection{Enumeration of hypermaps}

Each time a particular combinatorial problem is solved in terms of topological recursion, there occurs a natural question whether this leads to an interesting CohFT inside this combinatorial problem, and, as a consequence, to an interesting ELSV-type formula for it. This logic is explained in detail in~\cite[Introduction]{DLPS15}. In particular, the topological recursion was proved in~\cite{DOPS14} for the enumeration of \emph{hypermaps}, see also~\cite{DoManescu}. 

In the case of hypermaps the correspondence between LG and TR gives us immediately a full description of the Frobenius manifold structure behind this combinatorial problem; it is a particular simple example of a so-called Hurwitz Frobenius manifold. In the simplest case one can say that the Frobenius manifold with the prepotential $t_1^2t_2/2+t_2^2\log t_2$ resolves, via its associated CohFT and the ELSV-type formula, the combinatorial problem known, in different versions, as generalized Catalan numbers, discrete volumes of moduli spaces, or discrete surfaces~\cite{AlexMirMor,DMSS12,EO09,NorburyLattice}. This explains, in a conceptual way, some observations already made in~\cite{ACNP,FangLiuZong}.

\subsubsection{Bergman kernel and Torelli marking}  \label{sec:bergman}
Another important application of this paper is to prove a form of independence of the output of topological recursion from the choice of the bidifferential $B$ for a global spectral curve.  Topological recursion depends on $B$ and there are many ways to normalize $B$ depending on a choice of Torelli marking on the Riemann surface.  We show that for a global spectral curve satisfying a compatibility condition, topological recursion gives rise to a so-called homogeneous CohFT with flat identity {\em independent} of the choice of normalisation of $B$.

\subsection{Guide to the paper}

In Section~\ref{sec:Recollection} we give a full description of all concepts mentioned in Diagram~\eqref{eq:table-of-correspondences} and explain the known relations between them. 

In Section~\ref{sec:ycoinc} we prove that Dubrovin's superpotential
always gives the right $y$-function for the topological recursion. 
Then in Section~\ref{sec:Rcompat} we revisit in geometric terms the necessary compatibility conditions between $y$ and $B$ on the spectral curve from~\cite{DOSS12}. This allows us to prove the two main theorems of this paper.
Namely, in Section~\ref{sec:genus0} we prove the LG-TR correspondence in the genus 0 case, and in Section~\ref{sec:highergenus} we generalize this result to higher genera.

Then we discuss several important series of examples, where Dubrovin's superpotential can be computed explicitly. In Section~\ref{sec:Ansing} we discuss $A_n$ singularities, with an application to the Bouchard-Eynard generalisation of topological recursion. In Section~\ref{sec:hypermaps} we present in detail a computation for a special class of Hurwitz Frobenius manifolds, corresponding to the case of meromorphic functions on the Riemann sphere with two poles, one of which is of order 1. In this case the corresponding topological recursion resolves enumeration of hypermaps. 
In Section~\ref{sec:elliptic} we describe a higher genera case, namely, we consider the case of elliptic curve, where the superpotential is given by the Weierstrass function.

Section~\ref{sec:general-theory} is devoted to a general theory where we use a universal construction of analytic continuation instead of the rather particular constructions of Sections~\ref{sec:Ansing}, \ref{sec:hypermaps}, and~\ref{sec:elliptic}. This essentially reproduces, in our context, the main ideas of the work of Milanov~\cite{MilanovGlobal} initially applied by him to the case of simple singularities.

In the appendix we explicitly construct global spectral curves for two rank 2 CohFTs.  We need to vary the construction slightly due to degeneracy of the Gauss-Manin system.  These examples satisfy the conditions of Theorem~\ref{thm:any-genus-identification} and hence topological recursion produces the CohFT associated to the Frobenius manifold.

\subsection*{Acknowledgements}

We thank Todor Milanov for useful discussions and Boris Dubrovin for suggesting the connection between superpotentials and topological recursion.

A.~P. and S.~S. were supported by a Vici grant of the Netherlands Organization for Scientific research.  P.~D.-B. was partially supported by RFBR grants 16-31-60044-mol-a-dk, 15-01-05990, 15-31-20832-mol-a-ved and 15-52-50041-YaF-a and RFBR-India grant 14-01-92691-Ind-a; A.~P. was partially supported by RFBR grants 13-02-00457, 14-01-31492-mol-a, 15-52-50041-YaF-a and 15-31-20832-mol-a-ved. P.~D.-B. was also partially supported by the Government of the Russian Federation within the framework of a subsidy granted to the National Research University Higher School of Economics for the implementation of the Global Competitiveness Program and by an MPIM Bonn fellowship.

\section{Recollection of basic facts}
\label{sec:Recollection}

The purpose of this Section is to recall all necessary definitions and facts on Frobenius manifold, moduli spaces of curves, cohomological field theories, Dubrovin's universal construction of Landau-Ginburg superpotentials, and topological recursion. 

\subsection{Frobenius manifolds} 
\label{sec:Recollection-FrobeniusManifolds}
In this Section we recall, following~\cite{Dub2dTFT,Dub98}, the definition of Frobenius manifold and recollect some basic facts about its structures.

Consider a function $F(t^1,\dots,t^n)$ defined on a ball $B\subset \mathbb{C}^n$ and a constant inner product $\eta^{\alpha\beta}$ such that the triple derivatives of $F$ with one shifted index,
\begin{equation}
C_{\alpha\beta}^\gamma:=\frac{\partial^3 F}{\partial t^\alpha \partial t^\beta \partial t^\lambda} \eta^{\lambda\gamma},
\end{equation}
are the structure constants of a commutative associative Frobenius algebra with the scalar product given by $\eta_{\alpha\beta}$. We can think about this structure as defined on the tangent bundle of $B\subset \mathbb{C}^n$ (and we denote the corresponding multiplication of vector field by $\cdot$), and we require that $\partial_{t^1}$ is the unit of the algebra in each fiber. 

Consider a vector field $E:=\sum_{\alpha=1}^n ((1-q_\alpha)t^\alpha+r_\alpha) \partial_{t^\alpha}$, here $q_\alpha$ and $r_\alpha$ are some constants, $\alpha=1,\dots,n$. We require that $q_1=0$ and $r_\alpha\not=0$ only in the case $1-q_\alpha=0$. We require that there exists a constant $d$ such that $E.F-(3-d)F$ is a polynomial of order at most $2$ in $t^1,\dots,t^n$.  

The triple $(F,\eta,E)$ that satisfies all conditions above gives us the structure of a (conformal) Frobenius manifold of rank $n$ and conformal dimension $d$. The function $F$ is called the prepotential; the vector field $E$ is called the Euler vector field. Of course, there are coordinate-free descriptions of this structure as well, we refer to~\cite{Dub2dTFT,Dub98} for details.

Two important structures associated to Frobenius manifolds are the second metric $\eta'$ on $TB$ and the extended flat connection $\tilde\nabla$ on $B\times \mathbb{C}$. The second metric $\eta'$ on $TB$ is defined in the following way. The first metric $\eta$ can be considered as an isomorphism between $\eta\colon TB\to T^*B$. For any two vector fields $\partial'$ and $\partial''$ we define $\eta'(\partial',\partial'')$ to be $E\vdash \eta(\partial'\cdot \partial'')$. The extended connection $\tilde{\nabla}$ is defined as 
\begin{align}
\tilde\nabla_{\partial'}\partial'' &:= \nabla^{\eta}_{\partial'}\partial'' + z \partial'\cdot \partial''; \\
\tilde\nabla_{\partial'}\partial_z &:= 0;\\
\tilde\nabla_{\partial_z}\partial_z &:= 0;\\ \label{eq:1stTimeMu}
\tilde\nabla_{\partial_z}\partial' &:= \partial_z(\partial')+E\cdot \partial' - \frac{1}{z}\mu\partial',
\end{align}
where $\nabla^{\eta}$ is the Levi-Civita connection of $\eta$, and the endomorphism $\mu:TB\to TB$ is defined by 
\begin{equation}  \label{mu}
\mu(v):=(1-d/2)v-\nabla^{\eta}_vE.
\end{equation}  
In the flat basis, $\mu=\diag(\mu_1,\dots,\mu_n)$ for constants $\mu_\alpha=q_\alpha-d/2$. 

In this paper we only consider semi-simple Frobenius manifolds, that is, we require that the algebra structure on an open subset $B^{ss}\subset B$ is semi-simple. In a neighborhood of a semi-simple point we have a system of canonical coordinates $u_1,\dots,u_n$, defined up to permutations, such that the vector fields $\partial_{u_i}$, $i=1,\dots,n$, are the idempotents of the algebra product, and the Euler vector field has the form $E=\sum_{i=1}^n u_i\partial_{u_i}$. 

The geometric structure that is equivalent to the notion of conformal Frobenius manifolds can be described in canonical coordinates~\cite{Dub2dTFT}. 
The canonical coordinate vector fields $\partial_{u_i}$ are orthogonal but not orthonormal. We can normalize them to produce a so-called normalized canonical frame in each tangent space, that is, if $\Delta_i^{-1}=\eta(\partial_{u_i},\partial_{u_i})$, then the orthonormal basis is given by $\Delta_i^{1/2} \partial_{u_i}$, $i=1,\dots,n$. By $\Psi$ we denote the transition matrix from the flat basis to the normalized canonical one.  Hence the columns of $\Psi$ are given by the coordinates of the flat vectors $\partial_{t_{\alpha}}$ in the basis $\Delta_i^{1/2} \partial_{u_i}$, with first column $\Psi_{i1}=\Delta_i^{-1/2}$ representing the unit vector.  We have the relation
$$ E\cdot\Psi=\Psi\mu
$$
where $E\cdot$ is differentiation with respect to $E$.


Define the matrix $V$ to be the endomorphism $\mu$ with respect to the normalized canonical basis, hence $V=\Psi\cdot\diag(\mu_1,...,\mu_n)\cdot\Psi^{-1}$ and $V+V^T=0$. Covariant constancy of $\mu$ implies that $V$ satisfies
$$ dV=[V,d\Psi\cdot\Psi^{-1}].
$$
Define $V_i=\partial_{u_i}\Psi\cdot\Psi^{-1}$ so $\sum_i u_iV_i=V$.
\begin{remark}
Note that Givental \cite{GiventalMain} (and \cite{DOSS12}) uses a different convention for matrices than what is used here.  Givental's convention uses a right action of matrices on vectors which is the transpose of the convention we use here.
\end{remark}

\subsection{Superpotential}
A convenient way to describe a Frobenius structure is in terms of a so-called Landau-Ginzburg superpotential. We recall the definition from~\cite{Dub2dTFT,Dub98}.  A {\em superpotential} is a function $\lambda(p,u_1,\dots,u_n)$ of a variable $p\in\mathcal{D}$ in some domain $\mathcal{D}$ that depends on points $(u_1,\dots,u_n)\in B_0\subset B^{ss}$ in a ball in the semisimple part of the Frobenius manifold, and satisfies the following properties:
\begin{enumerate}
	\item The critical values of $\lambda$ as a function on $\mathcal{D}$ are $u_1,\dots,u_n$. 
	\item The critical points are non-degenerate.
	\item If there are several critical points in the inverse image $\lambda^{-1}(u_i)$, then the Hessians of $\lambda$ at these points must coincide.
	\item 
For any choice $p_1,\dots,p_n\in\mathcal{D}$ of the critical preimages of $u_1,\dots,u_n$ (that is, $\lambda(p_i,u_1,\dots,u_n)=u_i$) and for any choice of the vector fields $\partial'$, $\partial''$,
	 and $\partial'''$ on $B_0$ we have:
\begin{align}
\eta(\partial',\partial'')&=-\sum_{i=1}^n \Res_{p\to p_i} \frac{\partial'(\lambda dp)\partial''(\lambda dp)}{d_p \lambda}; \label{spmet}\\
\eta'(\partial',\partial'')&=-\sum_{i=1}^n \Res_{p\to p_i} \frac{\partial'(\log\lambda dp)\partial''(\log\lambda dp)}{d_p \log\lambda}; \label{spint}\\
\eta(\partial'\cdot \partial'',\partial''') & = -\sum_{i=1}^n \Res_{p\to p_i} \frac{\partial'(\lambda dp)\partial''(\lambda dp)\partial'''(\lambda dp)}{dp\, d_p \lambda} \label{sp3pt}
\end{align}
where $\partial'(\lambda dp)$ gives the action of the vector field by derivation in the parameters $u_i$.
In particular, the map $\partial'\mapsto\partial'(\lambda dp)$ from vector fields on $B$ to meromorphic differentials on $\mathcal{D}$ quotiented out by $d_p\lambda$ is injective.
\item There exist some cycles $Z_1,\dots,Z_n$ in $\mathcal{D}$ such that the integrals
\begin{equation}
\frac{1}{\sqrt z} \int_{Z_\alpha} e^{z\lambda} dp, \qquad \alpha=1,\dots,n
\end{equation}
converge and give a non-degenerate system of flat coordinates for $\tilde{\nabla}$.
\end{enumerate}
In these terms, the identity vector field $\partial_0$ of the  Frobenius manifold is represented by $dp$, i.e. $\partial_0(\lambda dp)=dp$. Indeed, since 
$\eta(\partial_0\cdot \partial,\partial')=\eta(\partial,\partial')$ for all vector fields $\partial,\partial'$, then non-degeneracy of $\eta$ implies that $\partial_0\cdot \partial=\partial$ for all $\partial$.  The Euler vector field is represented in these terms by $\lambda dp$, i.e. $E(\lambda dp)=\lambda dp$.

\subsection{Cohomological field theories} 
In this Section we recall all basic definitions that are necessary to introduce the concept of a cohomological field theory. It is an algebraic structure on a given vector space that captures the main properties of Gromov-Witten theories, and there is a natural group action on these structures, due to Givental. The main sources for this Section are~\cite{KontsevichManin,Teleman,GiventalMain,Shadrin09,PandPixZvo}.

A stable curve of genus $g$ with $k$ marked points is a possible reducible curve with nodal singularities, of arithmetic genus $g$ and $k$ non-singular marked points, such that the group of its automorphisms is finite. 
By $\overline{\mathcal{M}}_{g,k}$ we denote the moduli space of stable curves of genus $g$ with $k$ ordered marked points. There are natural line bundles $L_i\to \overline{\mathcal{M}}_{g,k}$, $i=1,\dots,k$, whose fiber of the point $[(C_g,x_1,\dots,x_k)]\in \overline{\mathcal{M}}_{g,k}$ represented by the curve $C_g$ with the marked points $x_1,\dots,x_k\in C_g$ is given by $T^*_{x_i}C_g$. The first Chern class of $L_i$ is denoted by $\psi_i\in H^2(\overline{\mathcal{M}}_{g,k},\mathbb{C})$. 

There are a number of natural maps between the moduli spaces. By $\pi\colon \oM_{g,k+1}\to\oM_{g,k}$ we denote the map that forgets the last marked point and stabilizes the curve. By $\sigma\colon\oM_{g_1,k_1+1}\times \oM_{g_2,k_2+1}\to \oM_{g,k}$ we denote the map that sews the last marked points on the source curves into a node on the target curve, $g=g_1+g_2$, $k=k_1+k_2$. By $\rho\colon \oM_{g-1,k+2}\to\oM_{g,k}$ we denote the map that sews the two last marked points on the source curve into a node on the target curve. 

Consider a vector space $V=\mathbb{C}\langle e_1,\dots,e_n \rangle$ with a scalar product $\eta$. A cohomological field theory with the target $(V,\eta)$ is a system of cohomology classes $\alpha_{g,k}\colon V^{\otimes k}\to H^*(\overline{\mathcal{M}}_{g,k},\mathbb{C})$ satisfying the following conditions:
\begin{enumerate}
	\item The form $\alpha_{g,k}$, $g\geq 0$, $k\geq 0$, $2g-2+k>0$, is invariant under the action of $S_k$ that simultaneously reshuffle $V^{\otimes k}$ and relabel the marked points on the curves in $\overline{\mathcal{M}}_{g,k}$.
	\item We have:
	\begin{align}
	\pi^*\alpha_{g,k}& =e_1\vdash \alpha_{g,k+1}; \label{flatvac}\\
	\sigma^*\alpha_{g,k+1} &= \eta^{\alpha\beta}e_\alpha\otimes e_\beta \vdash \alpha_{g_1,k_1+1} \alpha_{g_2,k_2+1}; \\
	\rho^*\alpha_{g,k} &= \eta^{\alpha\beta}e_\alpha\otimes e_\beta \vdash \alpha_{g-1,k+2}.
	\end{align}
	Here by $\vdash$ we denote the substitution of the vector $e_1$ at the $(k+1)$-st argument in the first equation, and the substitution of the bivector corresponding to the scalar product at the marked points that are sewed into the nodes under the maps $\sigma$ and $\rho$. 
\end{enumerate}

Note that if all classes $\{\alpha_{g,k}\}$ are of degree $0$, then the structure that we get is called a topological field theory (TFT), and it is equivalent to a Frobenius algebra structure on $(V,\eta)$. 

Correlators, or ancestor invariants, of the CohFT are defined by: 
\beq \label{tau} 
\int_{\overline{\mathcal{M}}_{g,k}}\alpha_{g,k}(e_{\nu_1},...,e_{\nu_k})\cdot\prod_{j=1}^k\psi_j^{m_j}
\eeq
for $m_i\in\N$, $\{e_{\nu},_{\ \nu=1,...,N}\}\subset H$.

There is a group action on CohFTs with a fixed target space $(V,\eta)$. The group is the group of matrices $R(z)\in \mathrm{End}(V)\otimes \mathbb{C}[[z]]$ such that $R=\Id+O(z)$ and $R(z)R^*(-z)=\Id$. The action is defined as follows. The classes $\{\alpha'_{g,k}\}=R.\{\alpha_{g,k}\}$ are defined as the sums over so-called stable graphs. 

A stable graph is a graph with a set of vertices $V$, a set of edge $E$, and a set of unbounded edges (leaves) $L\sqcup D$. The vertices are labeled by non-negative integers, that is, we have a map $V\to \mathbb{Z}_{\geq 0}$, $v\mapsto g(v)$. The stability condition means that for each vertex $v$ of valency $k(v)$
 we require $2g(v)-2+k(v)>0$. We say that the stable graph $\Gamma$ has genus $g$ and $k$ leaves if $b_1(\Gamma)+\sum_{v\in V}g(v)=g$ and $|L|=k$. So, we allow an arbitrary number of unbounded leaves in $D$ (these leaves are called dilaton leaves), that is, the set of stable graphs of genus $g$ with $k$ leaves is infinite. The leaves in $L$ are labeled from $1$ to $k$. 

A stable graph $\Gamma$ gives us a map $f_\Gamma$ from the Cartesian product of the spaces $\oM_{g(v),k(v)}$, $v\in V$, to $\oM_{g,k}$. Namely, we associate to each vertex $v$ a curve of genus $g(v)$, and to all attached half-edges we associate the marked points on the curve. Then we first apply the maps $\pi$ on each space $\oM_{g(v),k(v)}$, $v\in V$, in order to forget all marked points corresponding to the dilaton leaves, and then we apply a sequence of maps $\sigma$ and $\rho$, indexed by the edges $E$ of the graph, such that each edge determines the sewing of the corresponding curves. 

We associate to a stable graph $\Gamma$ a map from $V^{\otimes k}$ to $\otimes_{v\in V} H^*(\oM_{g(v),k(v)},\mathbb{C})$. That is, a map from $e_{\alpha_1}\otimes\cdots\otimes e_{\alpha_k}$ to the following class. We decorate by $R^{-1}(\psi) e_{\alpha_i}$ the leaf labeled by $i$. We decorate each dilaton leaf by $-\psi(\Id-R^{-1}(\psi))e_1$. We decorate each edge by 
\begin{equation}
\left(\frac{\Id\otimes\Id - R^{-1}(\psi')\otimes R^{-1}(\psi'')}{\psi'+\psi''}\right)\eta^{\alpha\beta}e_\alpha\otimes e_\beta,
\end{equation}
where by $\psi'$ and $\psi''$ we denote the $\psi$-classes associated with the marked points that correspond to the ends of the edge. Each vertex $v$ is decorated by $\alpha_{g(v),k(v)}$ considered as an element of $(V^{*})^{\otimes k(v)}\otimes H^*(\oM_{g(v),k(v)},\mathbb{C})$. We contract the tensor product of the vectors corresponding to edges and leaves with the tensor product of covectors corresponding to the vertices according to the graph. This gives us a class $\alpha_\Gamma$ in $\otimes_{v\in V} H^*(\oM_{g(v),k(v)},\mathbb{C})$.

By definition, the class $\alpha'_{g,k}$ is given by $\sum_\Gamma (f_\Gamma)_* \alpha_\Gamma$, where the sum is taken over all stable graphs of genus $g$ with $k$ leaves. Though there is an infinite number of graphs like that, one can check that only a finite number of them can contribute to this sum for dimensional reasons. It is indeed a group action on CohFTs, see e.~g.~\cite{PandPixZvo}. 

There is a canonical way to associate a CohFT to a semi-simple point of a Frobenius manifold. Namely, we associate to a point $b\in B^{ss}$ of a Frobenius manifold the topological field theory $\{\alpha_{g,k}\}$ with values in $(T_bB, \eta|_b)$. The equations for the flat sections of the connection $\tilde{\nabla}$ has essential singularity at $z=\infty$. The asymp\-totic fundamental solution near $z=\infty$ can be represented in a neighborhood of $b$ as $\Psi^{-1} R(z^{-1}) e^{zU}$, where all involved matrices are functions on $B^{ss}$, and the matrix $R$ satisfies all properties required in the definition of the group action. We can construct a CohFT applying the group element $R(z)|_b$ to the topological field theory on $(T_bB, \eta|_b)$.

\subsection{Dubrovin's superpotential}
\label{sec:DubrovinsSuperpotential}

In this Section we recall a construction of a particular Landau-Ginzburg superpotential due to Dubrovin~\cite{Dub98}.  

Given a manifold $M$ equipped with a flat metric, a locally defined function $t$ is a {\em flat coordinate} at $p\in M$, if 
\begin{itemize}
\item[(i)] $dt(p)\neq 0$ and 
\item[(ii)] $dt$ is covariantly constant with respect to to the Levi-Civita connection.
\end{itemize}
Condition (i) guarantees that $t$ is a local coordinate, i.e. we can find a coordinate system $(t^1,...,t^n)$ with $t^1=t$ and an open neighbourhood $B\subset M$ of $p$ such that $(t^1,...,t^n):B\to B_0\subset\R^n$ is a homeomorphism onto an open set $B_0$ of $\R^n$.  Condition (ii), which uses the induced connection on the cotangent bundle, guarantees that $(t^1,...,t^n)$ can be chosen so that the metric is represented by a constant matrix with respect to $(t^1,...,t^n)$.

We now consider a flat coordinate $\rho(\lambda,u)$ with respect to the pencil of metrics $\eta'-\lambda\eta$.  We study covariant constancy of $d\rho$ 
via its gradient vector field $\phi(\lambda,u)=\nabla \rho(\lambda,u)$ defined by
$$ \left(\eta'-\lambda\eta\right) (\phi,\cdot)=d\rho.
$$
The Levi-Civita connection of $\eta'$ with respect to flat coordinates (for $\eta$) is given in \cite[Equation~(5.5)]{Dub98}).  This leads to the following system of equations for vector fields $\phi$ expressed in canonical coordinates on a Frobenius manifold (the extended Gauss-Manin system~\cite[Equations~(5.31) and~(5.32)]{Dub98}):
\begin{equation}\label{eq:Gauss-Manin-system-5-31-32}d\phi = -(U-\lambda)^{-1} d(U-\lambda) \left(\frac{1}{2} +V\right) \phi + d\Psi\cdot\Psi^{-1} \phi.\end{equation}
Here $d=d_\lambda+d_u$ is the total de Rham differential; $U=\diag(u_1,...,u_n)$ 
and $V$ and $\Psi$ are naturally associated to a Frobenius manifold as defined in Section~\ref{sec:Recollection-FrobeniusManifolds}.  Abusing notation, we use $\lambda$ for the matrix of multiplication by $\lambda$.  So \eqref{eq:Gauss-Manin-system-5-31-32} encodes the system of PDEs giving covariant constancy of $\phi(\lambda,u)=\nabla \rho(\lambda,u)$ in directions $\partial/\partial\lambda,\partial/\partial u_i$.


One can retrieve $\rho$ from its gradient vector field via
\beq \label{gradtocoord}
\rho(\lambda,u)=\frac{\sqrt 2}{1-d} \phi^T (U-\lambda) \Psi \un.
\eeq
This is proved in~\cite[Section 2]{DubAlmost}. 



This equation has poles at $\lambda=u_1,\dots,u_n$ on the $\lambda$-plane, so we choose parallel cuts $L_1,\dots,L_n$ from the points $u_i$ to infinity (we assume that $u_j\not\in L_i$ for $i\not= j$). On $\C\setminus \cup_{i=1}^n L_i$ we choose branches of functions $\sqrt{u_i-\lambda}$, $i=1,\dots,n$. We denote by 
$\refl_i$ the
monodromy of the space of solutions of Equation~\eqref{eq:Gauss-Manin-system-5-31-32} corresponding to following a small loop around $u_1$.

Dubrovin proves that there exist a unique system of solutions  $\phi^{(1)},\dots,\phi^{(n)}$ to equation \eqref{eq:Gauss-Manin-system-5-31-32} satisfying the following properties:
\begin{align} \label{eq:DubrovinProp-1}
& \refl_j\phi^{(j)}=-\phi^{(j)}, & j=1,\dots, n; \\
& \phi^{(j)}_j = \frac{1}{\sqrt{u_j-\lambda}} + O(\sqrt{u_j-\lambda})\ \text{for}\ \lambda\to u_j, & j=1,\dots, n; \\
& \phi^{(j)}_a = \sqrt{u_j-\lambda}\cdot O(1) \ \text{for}\ \lambda\to u_j, & a\not=j;\; a,j=1,\dots, n;\\
& \refl_j\phi^{(i)}=\phi^{(i)}-2G^{ij} \phi^{(j)}, & i,j=1,\dots,n;
\label{eq:DubrovinProp-4}
\end{align}
where $G^{ij} := (\phi^{(i)})^T (U-\lambda) \phi^{(j)}$ is a bilinear form that doesn't depend on $\lambda$ and $u_1,\dots,u_n$.

Assume that $G^{ij}$ is non-degenerate and denote by $G_{ij}$ the inverse matrix.  Note that non-degeneracy of $G^{ij}$ is a property of the Frobenius manifold $M$  which holds generically.  In fact the proof of Theorem~\ref{thm:right-function-y} does not require the non-degeneracy of $G^{ij}$---see Remark~\ref{rem:more-general-phi}.  Consider a special solution of Equation~\eqref{eq:Gauss-Manin-system-5-31-32} given by $\phi:=\sum_{i,j=1}^n G_{ij} \phi^{(j)}$. The main property of this solution is that $\phi$ has the local behavior 
\begin{align} \label{eq:phi-locally-eq1}
& \phi_j = \frac{1}{\sqrt{u_j-\lambda}} + O(1)\ \text{for}\ \lambda\to u_j, & j=1,\dots,n;\\
\label{eq:phi-locally-eq2}
& \phi_a = \sqrt{u_j-\lambda}\cdot O(1)\ \text{for}\ \lambda\to u_j, & a\not=j; a,j=1,\dots,n.
\end{align}

We consider the function $p=p(\lambda,u)$ given by the formula 
\begin{equation}\label{eq:Superpotential-Formula}
p(\lambda,u) := \frac{\sqrt 2}{1-d} \phi^T (U-\lambda) \Psi \un .
\end{equation}
This function is analytic in $\C\setminus \cup_{i=1}^n L_i$, with a regular singularity at infinity, and its local behavior for $\lambda\to u_i$ is given by 
\begin{equation}\label{eq:p-lambda-locally}
p(\lambda,u) = p(u_i,u)+ \Psi_{i,\un} \sqrt{2(u_i-\lambda)}+O(u_i-\lambda),\qquad i=1,\dots,n.
\end{equation}

The 1-form $d_\lambda p$ has at most a finite number of zeros. We denote them by $r_1,\dots,r_N$ and we assume that they do not belong to the cuts $L_i$, $i=1,\dots,n$. Let $D$ be the image of $\C\setminus \cup_{i=1}^n L_i$ under the map $p(\lambda,u)$. This domain has a boundary given by the unfolding of the cuts $L_i$, $i=1,\dots,n$. The inverse function $\lambda=\lambda(p,u)$ is a multivalued function on $D$. Consider the points $p(r_c,u)$, $c=1,\dots,N$. We glue a finite number of copies of $D$ along the cuts from the points $p(r_c,u)$ to infinity, $c=1,\dots,N$. In this way we obtain a domain $\hat D$, where the function $\lambda$ is single-valued. 

We analytically continue the function $\lambda$ on $\hat D$ beyond the boundary. This procedure is not unique; for instance, we can glue several copies of $\hat D$ along the boundaries that are the images of the same cuts on the $\lambda$-plane. In any case, we can perform this construction uniformly over a small ball in the space of parameters $u_1,\dots,u_n$. This way we obtain a (not necessarily compact) Riemann surface $\mathcal{D}$, with a function $\lambda=\lambda(\tilde p, u)\colon \mathcal{D}\to\mathbb{C}$ (by $\tilde p$ we denote some local coordinate on $\mathcal{D}$).

Dubrovin proves in~\cite{Dub98} that the family of functions $\lambda(\tilde p, u)$ defined this way is a superpotential of the Frobenius manifold which was the input of this construction.

\subsection{Spectral curve topological recursion} \label{sec:TR} In this Section, we recall the basic set-up of the topological recursion procedure, which originated in the computation of the correlation functions of matrix models~\cite{EO,Ey}.

Consider a Riemann surface $\Sigma$ with meromorphic functions $x,y\colon \Sigma\to \mathbb{C}$ such that $x$ has a finite number of critical points, $c_1,\dots,c_n$, and $y$ is holomorphic near these points with a non-vanishing derivative. Let $B$ be a symmetric bi-differential on $\Sigma\times\Sigma$, with a double pole on the diagonal, the double residue equal to $1$, and no further singularities. 

We define a sequence of symmetric $n$-forms $\omega_{g,k}(z_1,\dots,z_k)$ on $\Sigma^{\times k}$, known as {\em correlation differentials} for the spectral curve, by the following recursion:
\begin{align}
& \omega_{0,1}(z):=y(z) dx(z); \\
& \omega_{0,2}(z_1,z_2):=B(z_1,z_2);\\ \label{eq:TopologicalRecursion-1stTime}
& \omega_{g,k+1}(z_0,z_1,\dots,z_k) :=  \\ \notag
& \sum_{i=1}^n \Res_{z\to c_i} 
\frac{\int_z^{\sigma_i(z)}\omega_{0,2}(\bullet,z_0)}{2(\omega_{0,1}(\sigma_i(z))-\omega_{0,1}(z))}
\tilde\omega_{g,2|k}(z,\sigma_i(z)|z_1,\dots,z_k),
\end{align}
where $\sigma_i$ is the deck transformation for the function $x$ near the point $c_i$, $i=1,\dots,n$, and $\tilde\omega_{g,2|k}$ is defined by the following formula:
\begin{align}
\tilde\omega_{g,2|k}(z',z''|z_1,\dots,z_k):=& \omega_{g-1,n+2}(z',z'',z_1,\dots,z_k)+
\\ \notag &
\sum_{\substack{g_1+g_2=g\\ I_1\sqcup I_2 = \{1,\dots,k\} \\ 2g_1-1+|I_1|\geq 0 \\ 2g_2-1+|I_2|\geq 0}}
\omega_{g_1,|I_1|+1}(z',z_{I_1})\omega_{g_2,|I_2|+1}(z'',z_{I_2}).
\end{align}
Here we denote by $z_I$ the sequence $z_{i_1},\dots,z_{i_{|I|}}$ for $I=\{i_1,\dots,i_{|I|}\}$.

\begin{remark}  \label{y=prim}
In the global recursion we also allow $y$ to be the (multivalued) primitive of a differential $\omega$ on $\Sigma$.   The ambiguity in $y$ consists of periods and residues of $\omega$ and hence the ambiguity is locally constant.  Since $y$ appears in the recursion only via $y(\sigma_i(z))-y(z)$ (and there are no poles of $\omega$ at the zeros of $dx$) the locally constant ambiguity disappears and the recursion is well-defined.
\end{remark}

\begin{remark}
A local version of the recursion was  defined in \cite{EynardIntersection} as follows.  Consider some small neighborhoods $U_i\subset \Sigma$ of the points $c_i$. If we look at just the restrictions of $\omega_{g,k}$ to the products of these disks, $U_{i_1}\times \cdots\times U_{i_k}$, we can still proceed by topological recursion, using as an input the restrictions of $\omega_{0,1}$ to $U_i$, $i=1,\dots,n$, and $\omega_{0,2}$ to $U_i\times U_j$, $i,j=1,\dots,n$. Indeed, Equation~\eqref{eq:TopologicalRecursion-1stTime} uses only local data for the recursion. 
\end{remark}

\begin{remark}
	There is a variation of the usual (global) topological recursion that will also be important in this paper, especially in Section~\ref{sec:general-theory}. Namely, we can assume that there is more than one critical point in the fiber of the function $x$ over a critical value $u_i$. Then we require that the local behavior of the function $x$ near these points is the same (that is, the Hessians are the same), and in this case it is still possible to define a version of topological recursion, see Section~\ref{sec:general-theory}. Note that this more general critical behavior of the function $x$ is exactly the one that is allowed for the function $\lambda$ in the definition of the Landau-Ginzburg superpotential of a Frobenius manifold in Section~\ref{sec:Recollection-FrobeniusManifolds}.
\end{remark}

\subsection{Spectral curve topological recursion via CohFTs}\label{sec:TRCohFT} In this Section we recall a relation of the (local version of) spectral curve topological recursion to the Givental formulae for cohomological field theories obtained in~\cite{DOSS12}. A more convenient exposition is given in~\cite{LPSZ}, so we follow the presentation given there.

We choose the local coordinates $w_i$ in the domains $U_i$ such that $x|_{U_i}=-w_i^2/2+x(c_i)$, $i=1,\dots,n$.  The identification with the data of a CohFT then goes as follows:
\begin{align} \label{eq:Identification-ConstantsDelta}
&\Delta_i^{-\frac 12} = \frac{dy}{dw_i}(0); \\
\label{eq:Identification-MatrixR}
& R^{-1}(\zeta^{-1})^j_i = -\frac{1}{\sqrt{2\pi\zeta}}\int_{-\infty}^{\infty} \left. \frac{B(w_i,w_j)}{dw_i} \right|_{w_i=0} \cdot e^{(x(w_j)-x(c_j))\zeta}; \\
\label{eq:Identification-Dilaton}
& \sum_{k=1}^n (R^{-1}(\zeta^{-1}))^i_k \Delta_k^{-\frac 12} = \frac{\sqrt\zeta}{\sqrt{2\pi}}\int_{-\infty}^\infty dy(w_i)\cdot e^{(x(w_i)-x(c_i))\zeta}.
\end{align}
Note that Equation~\eqref{eq:Identification-ConstantsDelta} is in fact a consequence of Equation~\eqref{eq:Identification-Dilaton}.

There is an extra condition on the bi-differential $B$ that can be formulated as a requirement on decomposition of its Laplace transform as 
\begin{align} \label{eq:decomposition-doubleLaplace-B}
 & \frac{\sqrt{\zeta_1\zeta_2}}{{2\pi}}\int_{-\infty}^{\infty} \int_{-\infty}^{\infty} {B(w_i,w_j)} e^{(x(w_i)-x(c_i))\zeta_1+(x(w_j)-x(c_j))\zeta_2}
 \\ \notag
 & =
 \frac{\sum_{k=1}^n R^{-1}(\zeta_1^{-1})^i_k
 	R^{-1}(\zeta_2^{-1})^j_k}{\zeta_1^{-1}+\zeta_2^{-1}} .
\end{align}
This assumption is always satisfied if the curve is compact and the differential $dx$ is meromorphic.   This uses a general finite decomposition for $B(p,q)$ proven by Eynard in Appendix B of \cite{EynardIntersection} together with \eqref{eq:Identification-MatrixR}.


This data (the constants $\Delta_i^{-\frac 12}$ and the matrix $R^{-1}(\zeta^{-1})^j_i$) determine for us a semi-simple CohFT $\{\alpha_{g,k}\}$ with an $n$-dimensional space of primary fields $V:=\langle e_1,\dots,e_n \rangle$. The differentials $\omega_{g,k}$ can be written in terms of the auxiliary functions
\begin{equation}
\xi^i(z):=\int^z \left. \frac{B(w_i,\bullet)}{dw_i}\right|_{w_i=0}
\end{equation}
as
\begin{equation}\label{eq:CohFT-SpectralCurve-match}
\omega_{g,k}=\sum_{\substack{i_1,\dots,i_k \\ d_1,\dots,d_k}} \int_{\overline{\mathcal{M}}_{g,k}} \alpha_{g,k}(e_{i_1},\dots,e_{i_k}) \prod_{j=1}^k
\psi_j^{d_j} d\left(\left(\frac{d}{dx}\right)^{d_j} \xi^{i_j}\right).
\end{equation}
(These kind of formulas are typically of ELSV-type, see~\cite{DLPS15} for explanation.)
In terms of the underlying Frobenius manifold structure, the basis $e_1,\dots,e_n$ corresponds to the normalized canonical basis.



\section{Superpotential and function $y$} \label{sec:ycoinc}

The goal of this Section is to prove that Dubrovin's superpotential provides us with a Riemann surface with two functions, $x:=\lambda$ and $y:=p$, such that the local expansion of $y$ near the critical points of $x$ reproduces the unit vector at the point $(u_1,\dots,u_n)$ of the underlying Frobenius manifold as well as the value of the matrix $R^{-1}$ on the unit vector.  These two local properties of $y$ are precisely equivalent to the equations~\eqref{eq:Identification-ConstantsDelta} and~\eqref{eq:Identification-Dilaton}.

Consider Dubrovin's construction of a superpotential on the Riemann surface $\mathcal{D}$ described in Section~\ref{sec:DubrovinsSuperpotential}. It is associated to a Frobenius manifold with given constants $\Delta_i^{-\frac 12}$ and the matrix $R^{-1}(\zeta^{-1})^j_i$ at the point with canonical coordinates $u_1,\dots,u_n$. Consider the points $c_i=p(u_i,u)\in\mathcal{D}$. These points are the critical points of the function $x:=\lambda$. 

\begin{theorem}\label{thm:right-function-y}
Given a semi-simple Frobenius manifold $M$, and Dubrovin's construction of a superpotential $\mathcal{D}$ for $M$, define spectral curve data by $\Sigma=\mathcal{D}$, $x:=\lambda$, $y:=p$ (with $B$ yet to be defined).  Then equations~\eqref{eq:Identification-ConstantsDelta} and~\eqref{eq:Identification-Dilaton} are satisfied for the 
	constants $\Delta_i^{-\frac 12}$ and the matrix $R^{-1}(\zeta^{-1})^j_i$ associated to $M$.
\end{theorem}

\begin{proof} Let us prove the first statement, namely, Equation~\eqref{eq:Identification-ConstantsDelta} (though it is a corollary of Equation~\eqref{eq:Identification-Dilaton}, it is convenient to check it directly). Indeed, Equation~\eqref{eq:p-lambda-locally} states that near the points $c_i$ the function $p$ looks like 
$$
p=c_i+\Psi_{i,\un}(u) \sqrt{2(u_j-\lambda)} + O(u_j-\lambda).
$$
Therefore, the derivative of $p$ with respect to the local coordinate $w_i=\sqrt{2(u_i-\lambda)}$ at the point $c_i$  is equal to $\Psi_{i,\un}(u)=\Delta_i^{-\frac 12}$.  

Now we prove Equation~\eqref{eq:Identification-Dilaton}. We can assume that the contour of integration on the right hand side in Equation~\eqref{eq:Identification-Dilaton} is the image of $L_i$ under the map $p$. Then, 
\begin{equation}
\frac{\sqrt\zeta}{\sqrt{2\pi}}\int_{p(L_i)} dp\cdot e^{(\lambda-u_i)\zeta} = \frac{\sqrt\zeta}{\sqrt{2\pi}}\int_{p(L_i)} \frac{dp}{d\lambda}\cdot e^{(\lambda-u_i)\zeta}d\lambda.
\end{equation}
Here we treat $dp$ and $d\lambda$ as 1-forms defined on the surface $\mathcal{D}$.

Observe that from equation \eqref{eq:Gauss-Manin-system-5-31-32} we have
\begin{equation}
\frac{d\phi^T}{d\lambda} = \phi^T \left(\frac 12 - V\right)(U-\lambda)^{-1} .
\end{equation}
Therefore, using definition \eqref{eq:Superpotential-Formula}, we get
\begin{align}
\label{eq:dpdlambda}
\frac{dp}{d\lambda} & = \frac{d}{d\lambda} \frac{\sqrt 2}{1-d} \phi^T (U-\lambda) \Psi \un =  \frac{\sqrt 2}{1-d} \phi^T \left(\frac 12 - V\right) \Psi\un -
\frac{\sqrt 2}{1-d} \phi^T \Psi\un \\ \notag
& = \frac{\sqrt 2}{1-d} \phi^T \Psi \Psi^{-1} \left(-\frac 12 - V\right) \Psi\un = \frac{\sqrt 2}{1-d} \phi^T \Psi \left(-\frac 12 - \mu\right) \un 
\\ \notag 
& =
 \frac{\sqrt 2}{1-d} \phi^T \Psi \left(-\frac 12 + \frac d2 \right) \un = -\frac 1{\sqrt{2}}\phi^T\Psi\un.
\end{align}
(In this computation we used 
the fact that $\mu\un=(-d/2)\un$.)

Equations~\eqref{eq:phi-locally-eq1} and~\eqref{eq:phi-locally-eq2} imply that on the contour $p(L_i)$ the vector $\phi$ is equal to $\phi^{(i)}+E_i$, where $E_i$ is some holomorphic function of $(u_i-\lambda)$. Recall also that $(\Psi\un)_k=\Delta_k^{-\frac 12}$. Therefore,
\begin{equation}\label{eq:Thm-y-laststep}
\frac{\sqrt\zeta}{\sqrt{2\pi}}\int_{p(L_i)} \frac{dp}{d\lambda}\cdot e^{(\lambda-u_i)\zeta}d\lambda = \sum_{k=1}^n \Delta_k^{-\frac 12} \cdot \frac{-\sqrt\zeta}{2\sqrt{\pi}}\int_{p(L_i)} \phi^{(i)}_k\cdot e^{(\lambda-u_i)\zeta}d\lambda.
\end{equation}
Dubrovin shows in \cite[Proof of Lemma 5.4]{Dub98} that the second factor in this expression is $(R^{-1}(\zeta^{-1}))^i_k$. Thus the right hand side of Equation~\eqref{eq:Thm-y-laststep} coincides with the left hand side of Equation~\eqref{eq:Identification-Dilaton}. This completes the proof of the Theorem.
\end{proof}

\begin{remark}\label{rem:more-general-phi}
	Note that we have not used the specific formula for $\phi$ in the proof. We used only Equation~\eqref{eq:Gauss-Manin-system-5-31-32} and the fact that the local expansion of $\phi$ for $\lambda\to u_i$ coincides with the local expansion of $\phi^{(i)}$ up to some holomorphic non-branching term. Thus, if we have a solution for \eqref{eq:Gauss-Manin-system-5-31-32} satisfying this property, we can use it directly in the formula for the superpotential~\eqref{eq:Superpotential-Formula}, bypassing the requirement for $G^{ij}$ to be non-degenerate. This will be important below in certain applications.
\end{remark}

\begin{remark} {\em Flat identity.}
Topological recursion satisfies the string equation.
\begin{equation}  \label{string}
\sum_{i=1}^n \res_{p=c_i}y(p)\omega_{g,k+1}(p,p_1,...,p_k)=-\sum_{j=1}^k d_{p_j}{\partial z_j}\left(\frac{\omega_{g,k}(p_1,...,p_k)}{dx(p_j)}\right)
\end{equation}
where the sum is over the zeros $dx(c_i)=0$ and $d_{p_j}$ is exterior derivative in the variables $p_j$.  The operator $\omega\mapsto\sum_i\mathrm{Res}_{p=c_i}y(p)\omega(p)$ acts on differentials $\omega$.  It is non-zero (and evaluates to 1) on the auxiliary differential $\sum_ja_j d \xi^j$ corresponding to the flat identity and annihilates all others.  In particular
$$
\sum_i\res_{p=c_i}d\left(\left(\frac{d}{dx}\right)^{d_j} \xi^{i_j}\right)=0,\quad d_j>0.
$$
This corresponds to insertion/removal of the identity vector in ancestor invariants.
\end{remark}



\section{Compatibility between $B$ and $y$} \label{sec:Rcompat}

In this section we discuss a necessary condition on a spectral curve to be able to apply the inverse construction of \cite{DOSS12}, i.e. so that a CohFT can be reconstructed from this spectral curve.

More precisely, for a given data of a spectral curve $(\Sigma,x,y,B)$ (maybe, local) Equations~\eqref{eq:Identification-MatrixR} and~\eqref{eq:Identification-Dilaton}, \eqref{eq:Identification-ConstantsDelta} imply some relation for $x$, $y$, and $B$, and we want to state this relation in a direct geometric way rather than in terms of the Laplace transform.

	The compatibility condition below is equivalent to differentiation of the potential of a CohFT by $t_1$ producing the string equation.  In the language of \cite{EO}, $\delta(ydx)=\int (dy/dx)(p')\, B(p,p')= d(dy/dx)$ gives rise to variations of $\omega_{g,k}$ corresponding to the string equation \eqref{string}.

Recall that $x$ defines a local involution $\sigma_i$ near each zero $c_i$ of $dx$, $i=1,\dots,n$.

\begin{theorem}
	\label{thm:CompatibilityTest} 
If a CohFT can be reconstructed from a spectral curve $(\Sigma,x,y,B)$ via the inverse construction of \cite{DOSS12} described in Section~\ref{sec:TRCohFT}, then the 1-form on $\Sigma$
\begin{equation}  \label{dosstest}
\eta(z)=d\left(\frac{dy}{dx}(z)\right)+\sum_{i=1}^n \res_{z'=c_i}\frac{dy}{dx}(z')B(z,z').
\end{equation}
is invariant under each local involution $\sigma_i$, $i=1,\dots,n$.  \end{theorem}
	

\begin{proof} The construction of \cite{DOSS12} requires equations~\eqref{eq:Identification-ConstantsDelta}, \eqref{eq:Identification-MatrixR}, and~\eqref{eq:Identification-Dilaton} to hold.  We will prove that the 1-form \eqref{dosstest}  is invariant under each local involution $\sigma_i$, $i=1,\dots,n$ if and only if equations~\eqref{eq:Identification-ConstantsDelta}, \eqref{eq:Identification-MatrixR}, and~\eqref{eq:Identification-Dilaton} are compatible (as equations for the unknown variables $R^{-1}$ and $\Delta_i^{-\frac 12}$, $i=1,\dots,n$).

Recall that $x=x(c_i)-w_i^2/2$ in a neighborhood of $c_i$. Note that 
\begin{align}
& \res_{w_i=c_i}\frac{dy}{dx}(w_i)B(z,w_i) =
\res_{w_i=c_i}\frac{dy}{dw_i}(w_i)\cdot \frac{dw_i}{dx}\cdot  B(z,w_i) \\ \notag
& = -\res_{w_i=c_i} \frac{dy}{dw_i}(w_i)\cdot \frac{dw_i}{w_i}\cdot  \frac{B(z,w_i)}{dw_i}
= \frac{dy}{dw_i}(0) \cdot \left. \frac{B(z,w_i)}{dw_i} \right|_{w_i=0} .
\end{align}	
	An equivalent way to say that $\eta$ is $\sigma_i$-invariant is to say that the following Laplace transform of $\eta$ is equal to zero:
\begin{equation}\label{eq:invariance-eta}
\int_{-\infty}^\infty \eta(w_i) e^{(x(w_i)-x(c_i))\zeta} = 0 .
\end{equation}	
On the other hand,
\begin{align}
 \int_{-\infty}^\infty \eta(w_i) e^{(x(w_i)-x(c_i))\zeta} 
 = 
& -\zeta \int_{-\infty}^\infty \frac{dy}{dx}(w_i)
 e^{(x(w_i)-x(c_i))\zeta} dx
 \\ \notag
 &
- \sum_{j=1}^n \frac{dy}{dw_j}(0) \int_{-\infty}^\infty \left. \frac{B(w_i,w_j)}{dw_j} \right|_{w_j=0} e^{(x(w_i)-x(c_i))\zeta}.
\end{align}	
Thus, Equation~\eqref{eq:invariance-eta} is satisfied if and only if 
\begin{align}
&\frac{\sqrt\zeta}{\sqrt {2\pi}} \int_{-\infty}^\infty {dy}(w_i)
e^{(x(w_i)-x(c_i))\zeta}
\\ \notag
&= \sum_{j=1}^n \frac{dy}{dw_j}(0) \cdot \frac{-1}{\sqrt {2\pi\zeta}} \int_{-\infty}^\infty \left. \frac{B(w_i,w_j)}{dw_j} \right|_{w_j=0} e^{(x(w_i)-x(c_i))\zeta},
\end{align}	
which is precisely the compatibility condition for Equations~\eqref{eq:Identification-ConstantsDelta}, \eqref{eq:Identification-MatrixR}, and \eqref{eq:Identification-Dilaton}.
\end{proof}

We can state \eqref{dosstest} in simpler terms when the spectral curve is connected.
\begin{corollary}\label{cor:dosstestpullback}
For a connected spectral curve, equations~\eqref{eq:Identification-ConstantsDelta}, \eqref{eq:Identification-MatrixR}, and~\eqref{eq:Identification-Dilaton} are compatible if and only if the 1-form defined in \eqref{dosstest} is a pull-back of a 1-form downstairs, i.e. $\eta(z)=x^*\omega$. 
\end{corollary}
\begin{proof}
If $\eta(z)=x^*\omega$ for $\omega$ a differential downstairs then it is invariant under local involutions hence Theorem~\ref{thm:CompatibilityTest} applies.  On a comnnected spectral curve $\Sigma$ the converse is also true.  This follows from the more general fact that any $\eta(z)$ which is invariant under local involutions defined around simple ramification points of $x:\Sigma\to\C$ is the pull-back of a differential downstairs.  Take any regular point of $x$ $p\in\Sigma$ and a path $\gamma$ from $p$ to a zero $b$ of $dx$.  Then $x(\gamma)$ is covered by a path $\tilde{\gamma}\subset\Sigma$ that contains $p$ and $p'$ where $x(p)=x(p')$.  The local involution defined by $x$ in a neighbourhood of $b$ can be analytically continued along $\tilde{\gamma}$.  Since $\eta(z)$ is invariant under the local involution at $b$, it is invariant under the continued involution above a neighbourhood of $x(\gamma)$.  So $\eta(z)$ agrees (via identification of cotangent bundles using $x$ ) around $p$ and $p'$.  Connectedness of $\Sigma$ guarantees that the monodromy of the cover defined by $x$ is transitive and generated by local involutions.  Hence we can find paths $\gamma_i$ that can be used to show that $\eta(z)$ agrees around $p$ and any point in the fibre over $x(p)$.  Hence $\eta(z)=x^*\omega$ locally and this pieces together to give the global result. The result isn't true on disconnected curves, in particular local curves, because monodromy is not transitive.
\end{proof}





Let us show how this compatibility test can be used. 

\begin{proposition}  \label{th:yprim}
The differential $\eta\equiv 0$, hence Equation~\eqref{eq:invariance-eta} is satisfied, when $\Sigma$ is a global curve equipped with a canonical bidifferential $B$ normalized so that $\int_{p'\in\alpha_i}B(p,p')=0$ for a choice of $A$-cycles $\alpha_i$, and one of the following holds:
\begin{enumerate}
\item $\Sigma$ is rational with global coordinate $z$ chosen so that $x(z=\infty)=\infty$; 
\item $dy$ is a meromorphic differential such that $\frac{dy}{dx}$ has poles only at the zeros of $dx$, for example $dy$ is a holomorphic differential.  \label{dyhol}
\end{enumerate}
\end{proposition}
Note that in case \eqref{dyhol} above, we take $y$ to be the (multiply-defined) primitive of a differential which is sufficient for the purposes of topological recursion---see Remark~\eqref{y=prim}.
\begin{proof}
Recall the property that for any function $f$ on $\Sigma$, ${\rm Res}_{p'=p}f(p')B(p,p')=df(p)$ (independent of the choice of $A$-cycles along which $B$ is normalized).  For example, in the rational case $B=\frac{dzdz'}{(z-z')^2}$  and this property is the Cauchy integral formula.  Since $\frac{dy}{dx}$ has poles only at the zeros of $dx$
$$\sum_{i=1}^n \res_{p'=c_i}\frac{dy}{dx}(p')B(p,p')=-\res_{p'=p}\frac{dy}{dx}(p')B(p,p')=-d\left(\frac{dy}{dx}(p)\right)$$
hence $\eta\equiv 0$.
\end{proof}

\begin{example}
Consider $x=z+1/z$, $y=p(z)$ a polynomial.  Then $\frac{dy}{dx}=\frac{z^2p'(z)}{z^2-1}$ has poles at $z=\pm1$ and possibly $z=\infty$.  Hence $\eta(z)=dq(z)$ where $q(z)$ is a polynomial given by the principal part of $dy/dx$ at $z=\infty$.  A non-trivial polynomial has poles only at $z=\infty$ so if $\eta\neq 0$ it cannot be the pull-back of a differential form downstairs since it would necessarily require poles at $x^{-1}(\infty)=\{0,\infty\}$.  Hence this fails the compatibility test, unless $\eta(z)\equiv 0$ i.e. $\deg p(z)\leq 1$.  If $\deg p(z)=1$ then Equation~\eqref{eq:invariance-eta} is satisfied.
\end{example}

\begin{example}
Consider $x=z+1/z$, $y=\ln{z}$.  Then $\frac{dy}{dx}=\frac{z}{z^2-1}$ has poles only at $z=\pm1$ so Equation~\eqref{eq:invariance-eta} is satisfied.
\end{example}

\begin{example}
Since the compatibility test is a linear condition in $y$, $x=z+1/z$, $y=\ln z+cz$ also satisfies the compatibility test and leads to a CohFT with a flat unit. This was also observed in~\cite{FangLiuZong}.
\end{example}



\section{Superpotential as a global spectral curve in genus $0$ case}

\label{sec:genus0}

In this Section we discuss a special case of Dubrovin's superpotential defined in Section~\ref{sec:DubrovinsSuperpotential} and show that it indeed gives a proper spectral curve for the correponding cohomological field theory.

More precisely, we start with a homogeneous cohomological field theory. Its genus zero part without descendants defines a Frobenius manifold 
that we assume to be semi-simple. 
%
%
Consider Dubrovin's construction in Section~\ref{sec:DubrovinsSuperpotential}. Assume that this construction goes through in such a way that 
\begin{enumerate}
	\item The form $d_\lambda p$ has no zeros in $\mathbb{C}\setminus \cup_{i=1}^n L_i$ \label{cond1}; 
\item $\lambda(p=\infty)=\infty$;
	\item The resulting curve $\mathcal{D}$ is a compact curve of genus $0$ and $p$ is a global coordinate on it;
\item
There is exactly one critical point in each singular fiber of function $\lambda$. \label{cond4} 
\end{enumerate}


\begin{theorem} 
	\label{thm:genus-0-identification}
	Under the conditions \eqref{cond1}-\eqref{cond4} above, 
	the correlators of the CohFT are related by Equation~\eqref{eq:CohFT-SpectralCurve-match} to the correlator differentials obtained through spectral curve topological recursion on a curve $\mathcal{D}$ with $x=\la$, $y=p$ and $B(p_1,p_2)=dp_1dp_2/(p_1-p_2)^2$.
\end{theorem}

In other words, in this case the ancestor potential of CohFT is reproduced by global topological recursion related to Dubrovin's superpotential. Note that this identification happens over an open ball in the underlying Frobenius manifold.

\begin{proof}
First of all, note that since $p$ is a global coordinate and $\lambda(p=\infty)=\infty$, this spectral curve satisfies the compatibility condition of Theorem~\ref{thm:CompatibilityTest}, which means that one can reconstruct a CohFT 
such that Equation~\eqref{eq:CohFT-SpectralCurve-match} is satisfied. We only need to prove that this CohFT is the same as the original one. 

Theorem~\ref{thm:right-function-y} implies that we have the right function $y$, so, in particular, the functions $\Delta_i^{-\frac 12}(u)$ are correctly reproduced on an open ball in the space of parameters $u_1,\dots,u_n$. Note that these functions determine completely the structure of Frobenius multiplication, so we can conclude that the CohFT reconstructed from the spectral curve data coincides with the original one in genus zero. 

Higher genera correlators of a semi-simple CohFT are determined uniquely by genus $0$ data in homogeneous cases~\cite{Teleman}. Therefore, it is sufficient to prove that the CohFT reconstructed from the spectral curve data is homogeneous. 
%
%
We do this by proving the Euler equation for the corresponding $R$-matrix. Namely, a CohFT with an $R$-matrix $R(\xi)$ is homogeneous if and only if the $R$-matrix satisfies the Euler equation~\cite{GiventalMain}:
\beq
\br{\xi \dif{}{\xi} + \sum_{i=1}^{n}u_i\pdif{}{u_i}} R(\xi,u) = 0 
\eeq
(or, equivalently, we can consider the same equation for $R^{-1}(\xi,u)=R(-\xi,u)^T$). Using  Equation~\eqref{eq:decomposition-doubleLaplace-B}, the Euler equation for the $R$-matrix can be rewritten as 
\beq \label{eq:homogeneity-B-check}
		\br{1 + \xi_1\pdif{}{\xi_1} + \xi_2\pdif{}{\xi_2} + \sum_{i=1}^n u_i\pdif{}{u_i}}\Bc=0
\eeq
for $\Bc=\Bc^{ij}(\xi_1,\xi_2)$ given by
\beq
\frac{e^{-\frac{u_i}{\xi_1}-
		\frac{u_j}{\xi_2}}}{{2\pi\sqrt{\xi_1\xi_2}}}
	\int_{p(L_i)} \int_{p(L_j)} B\cdot 
e^{\frac
		{\lambda_1}{\xi_1}+\frac{\lambda_2}{\xi_2}} .
\eeq

Recall that we consider the case when $d_\lambda p$ does not have zeros in $\mathbb{C}\setminus \cup_{i=1}^n L_i$, and the Riemann surface $\mathcal{D}$ that we get through Dubrovin's construction has genus 0. The Bergman kernel $B(p_1,p_2)$ has the form
$dp_1dp_2/(p_1-p_2)^2$.

\begin{proposition} \label{prop:homogeneirty}
	Under these conditions Equation~\eqref{eq:homogeneity-B-check} is satisfied.
\end{proposition}

We prove this proposition below. It implies that the $R$-matrix associated to the Bergman kernel in this case satsifies the Euler equation, and, therefore, the corresponding CohFT is homogeneous. This proposition completes the proof of Theorem~\ref{thm:genus-0-identification}. 
\end{proof}

For the proof of Proposition~\ref{prop:homogeneirty} we need the following technical lemma:
\begin{lemma} We have:
	\beq\label{eq:phomo}
	\lb \lambda\dif{}{\lambda} +\sum_{i=1}^{n} u_i \pdif{}{u_i} \rb p(\la,u) = \dfrac{1-d}{2} p(\la,u).
	\eeq
\end{lemma}
\begin{proof}
	Recall Equation~\eqref{eq:dpdlambda}:
	\beq
	d_\la p(\la,u) = -\dfrac{1}{\sqrt{2}} \phi^T d\lambda \Psi\un .
	\eeq
	In the same way we prove that 
	\beq
	d_u p(\la,u) = \dfrac{1}{\sqrt{2}}\phi^T dU \Psi \un 
	\eeq
	(this is \cite[equation (5.66)]{Dub98}; note that there is a misprint in this equation in~\cite{Dub98}). 
	Combining these equations, we get
	\beq
	\lb \lambda\dif{}{\lambda} +\sum_{i=1}^{n} u_i \pdif{}{u_i} \rb p(\la,u) = \dfrac{1}{\sqrt{2}}\phi^T (U-\lambda) \Psi \un = \frac{1-d}{2} p(\lambda,u).
	\eeq
\end{proof}

\begin{proof}[Proof of Proposition~\ref{prop:homogeneirty}]
We have:
\begin{align} \label{eq:expression-Euler-Bcheck}
&
	\br{1 + \xi_1\pdif{}{\xi_1} + \xi_2\pdif{}{\xi_2} + \sum_{i=1}^n u_i\pdif{}{u_i}}\Bc
	\\ \notag & 
	=\dfrac{e^{-\frac{u_i}{\xi_1}- \frac{u_j}{\xi_2}}}{2\pi\sqrt{\xi_1\xi_2}}\iint	
	 \dfrac{d\la_1d\la_2}{(p(\la_1)-p(\la_2))^2} \dif{p}{\la}\br{\la_1}\dif{p}{\la}\br{\la_2} e^{\frac{\la_1}{\xi_1}+\frac{\la_2}{\xi_2}}  X, 
\end{align}
where
\begin{align*}
	X = & -\frac{\la_1}{\xi_1}-\frac{\la_2}{\xi_2}
	-2\cdot \dfrac{\left(\sum\limits_{i=1}^n u_i\pdif{}{u_i}\right)\br{p(\la_1)-p(\la_2)}}{p(\la_1)-p(\la_2)} \\
	&  +\dfrac{\left(\sum\limits_{i=1}^n u_i\pdif{}{u_i}\right)\dif{p}{\la}\br{\la_1}}{\dif{p}{\la}\br{\la_1}}
	+\frac{\left(\sum\limits_{i=1}^n u_i\pdif{}{u_i}\right)\dif{p}{\la}\br{\la_2}}{\dif{p}{\la}\br{\la_2}} .
\end{align*}
Applying the integration by parts to the terms $-\lambda_1/\xi_1$ and $-\lambda_2/\xi_2$, we can rewrite the right hand side of Equation~\eqref{eq:expression-Euler-Bcheck} as 
\begin{equation}
\dfrac{e^{-\frac{u_i}{\xi_1}- \frac{u_j}{\xi_2}}}{2\pi\sqrt{\xi_1\xi_2}}\iint	
\dfrac{d\la_1d\la_2}{(p(\la_1)-p(\la_2))^2} \dif{p}{\la}\br{\la_1}\dif{p}{\la}\br{\la_2} e^{\frac{\la_1}{\xi_1}+\frac{\la_2}{\xi_2}}  Y,
\end{equation}
where
\begin{align*}
Y =\ & 2+
\dfrac{\br{\la_1\dif{}{\la_1}+\sum\limits_{i=1}^n u_i\pdif{}{u_i}}\dif{p}{\la}\br{\la_1}}{\dif{p}{\la}\br{\la_1}}
	+\dfrac{\br{\la_2\dif{}{\la_2}+\sum\limits_{i=1}^n u_i\pdif{}{u_i}}\dif{p}{\la}\br{\la_2}}{\dif{p}{\la}\br{\la_2}}-\\
	&-2\cdot \dfrac{\br{\la_1\dif{}{\la_1}+\sum\limits_{i=1}^n u_i\pdif{}{u_i}} p(\la_1)-\br{\la_2\dif{}{\la_2}+\sum\limits_{i=1}^n u_i\pdif{}{u_i}}p(\la_2)}{p(\la_1)-p(\la_2)} .
\end{align*}
Using Equation~\eqref{eq:phomo}, we rewrite $Y$ as
\begin{align*}
Y =\ & 2+	\dfrac{\br{-1+\dfrac{1-d}{2}}\dif{p}{\la}\br{\la_1}}{\dif{p}{\la}\br{\la_1}}
	+\dfrac{\br{-1+\dfrac{1-d}{2}}\dif{p}{\la}\br{\la_2}}{\dif{p}{\la}\br{\la_2}}-\\
	&-2\cdot \dfrac{\dfrac{1-d}{2} p(\la_1)-\dfrac{1-d}{2}p(\la_2)}{p(\la_1)-p(\la_2)}\\
=\	& 2+\left(-1+\dfrac{1-d}{2}\right)+\left(-1+\dfrac{1-d}{2}\right)-2\cdot \dfrac{1-d}{2}=0,
\end{align*}
which proves the proposition.
\end{proof}



\section{Superpotential as a global spectral curve for arbitrary genus}

\label{sec:highergenus}

In this Section we extend the result of the previous section to the case of a compact global curve of arbitrary genus.



\begin{theorem} 
	\label{thm:any-genus-identification}
Given a conformal Frobenius manifold, construct a superpotential $p(\lambda;u)$ which defines the Riemann surface ${\cal D}$ according to Dubrovin's construction of Section~\ref{sec:DubrovinsSuperpotential}.  
 Assume the following:
\begin{itemize}
\item  ${\cal D}$ is a compact curve of genus $g$;
\item there is exactly one critical point in each singular fiber of $\lambda:{\cal D}\to\mathbb{C}$.
\end{itemize}
Fix a symplectic basis $\left({\cal A}_i,{\cal B}_i\right)_{i=1}^g$ of $H_1({ \cal D}, \mathbb{Z})$ and define $B(p_1,p_2)$ as the only Bergman kernel on ${\cal D}$ normalized by 
\beq  \label{normBerg} 
\forall i=1,\dots,g \, , \; \oint_{p_1 \in {\cal A}_i} B(p_1,p_2) = 0.
\eeq	
Further assume that:
\begin{itemize}
\item the pair $(p, B(p_1,p_2))$ passes the compatibility test of Section~\ref{sec:Rcompat} in any of its possible forms (given by Theorem~\ref{thm:CompatibilityTest}, Corollary~\ref{cor:dosstestpullback}, or Proposition~\ref{th:yprim}).
\end{itemize}
Then the correlators of the CohFT associated to the Frobenius manifold are related by Equation~\eqref{eq:CohFT-SpectralCurve-match} to the correlator differentials obtained through spectral curve topological recursion on the Riemann surface $\mathcal{D}$ with $x=\la$, $y=p$ and $B(p_1,p_2)$. 
\end{theorem}

\begin{remark}

This result extends Theorem~\ref{thm:genus-0-identification} to an arbitrary compact curve. The new feature is that one needs to normalize the Bergman kernel on an arbitrary basis of cycles.  In particular, for each basis, we recover a total ancestor potential for the same CohFT. 


\end{remark}

\begin{proof}
The proof is very similar to the proof of the genus 0 case presented in the preceding section. However, it is important to remark that this proof only relies on Rauch's variational formula, i.e. it is valid for any compact curve presented as a ramified cover of the Riemann sphere with simple branch points. It does not require any knowledge about an auxiliary meromorphic form such as the super-potential.

Let us first show that the $(0,3)$ correlators are independent of choice of normalisation cycles for $B$.  $\omega_{0,3}$ depends on these choices, but when decomposed into linear combinations of auxiliary differentials $d\xi^j=B/ds_j$ (for $s_j$ defined by $x=(1/2)s_j^2+a_j$) the coefficients are independent of $A$-cycles. 
	By reconstruction, as in the proof of Theorem~\ref{thm:genus-0-identification}, this means that all correlators are the same. 
The formula 
	\begin{align*}
	\omega_{0,3}(z_1,z_2,z_3) &= \sum_{i=1}^n \Res_{p=c_i}  B(p,z_1)B(p,z_2)B(p,z_3)/dx(p)dy(p)\\
	&=\sum_{i=1}^n B(a_i,z_1)B(a_i,z_2)B(a_i,z_3)/x''(a_i)y'(a_i)\\
	&=\sum_{i=1}^n \langle....\rangle d\xi^i(z_1)d\xi^i(z_2)d\xi^i(z_3)
	\end{align*}
	shows the independence of the coefficients $\langle ...\rangle$ on the choice of $B$.

For the rest of the proof, the only part differing from the genus 0 case is the proof of the homogeneity of the CohFT, i.e. the fact that the $R$-matrix satisfies the Euler equation.

The first step consists in proving that there exist a $R$-matrix. This is due to a lemma of Eynard~\cite{EynardIntersection}:

\bl
If $d\lambda$ is a meromorphic form on ${\cal D}$ and $B$ the Bergman kernel normalized on a basis of ${\cal A}$-cycles as above, then the Laplace transform of the Bergman kernel satisfies Equation~\eqref{eq:decomposition-doubleLaplace-B} .

\el

The Euler equation for the $R$-matrix is then equivalent to the following equation for the Laplace transform of $B$:
\beq \label{eq:homogeneity-B-check-higher-g}
		\br{1 + \xi_1\pdif{}{\xi_1} + \xi_2\pdif{}{\xi_2} + \sum_{i=1}^n u_i\pdif{}{u_i}}\Bc=0
\eeq
for $\Bc=\Bc^{ij}(\xi_1,\xi_2)$ given by
\beq
\frac{e^{-\frac{u_i}{\xi_1}-
		\frac{u_j}{\xi_2}}}{{2\pi\sqrt{\xi_1\xi_2}}}
	\int_{p(L_i)} \int_{p(L_j)} B\cdot 
e^{\frac
		{\lambda_1}{\xi_1}+\frac{\lambda_2}{\xi_2}} .
\eeq
By inverting the Laplace transform and integration by part, this is equivalent to
\beq\label{eqhomogeneityB}
d_1\left({\lambda_1 \, B(p_1,p_2) \over d\lambda_1}\right) + d_2\left({\lambda_2 \, B(p_1,p_2) \over d\lambda_2}\right) + \sum_{i=1}^n u_i {\partial \over \partial u_i} B(p_1,p_2).
\eeq

In order to prove this equation, we remind Rauch's variational formula which expresses the variations of the Bergman kernel under deformation of the spectral curve. In particular
\beq\label{RauchB}
{\partial B(p_1,p_2) \over \partial u_i} = \Res_{r \to a_i} {B(p_1,r) \, B(p_2,r) \over d\lambda(r)}
\eeq
which implies that
\beq
\sum_{i =1}^n u_i {\partial \over \partial u_i} B(p_1,p_2)  = \sum_{i=1}^n \Res_{r \to a_i} {\lambda(r) B(p_1,r) \, B(p_2,r) \over d\lambda(r)} .
\eeq
Moving the integration contours around the other poles of the integrands and reminding that the ${\cal A}$-periods of $B(p,r)$ are vanishing, this reads
\beq
\sum_{i =1}^n u_i {\partial \over \partial u_i} B(p_1,p_2)  =   - \Res_{r \to p_1,p_2} {\lambda(r) B(p_1,r) \, B(p_2,r) \over d\lambda(r)}   = - \left( {d \over d \lambda_1} + {d \over d \lambda_2} \right) B(p_1,p_2) 
\eeq
proving Equation~\ref{eqhomogeneityB}.

\end{proof}





\section{Global curves for $A_n$ singularities}
\label{sec:Ansing}

In this Section we apply the results of Sections~\ref{sec:ycoinc}, \ref{sec:Rcompat}, and \ref{sec:genus0} in order to construct the spectral curve for the ancestor potential of $A_n$-singularities, $n=1,2,\dots$. The structure of this Frobenius manifolds is described in terms of Saito's theory on the space of polynomials 
\begin{equation}\label{eq:polynomial}
f(p,\tau)=p^{n+1}+\tau_1 p^{n-1}+\cdots \tau_n.
\end{equation}
 We refer to~\cite{Dub2dTFT,GiventalAn} for the detailed description of the structure of this Frobenius manifold. In particular, it is enough to say that $\lambda=f(p,\tau)$ is a superpotential of this Frobenius manifold.
 
The corresponding CohFT is well-studied. It was a subject of Witten's conjecture~\cite{WittenConj} proved in~\cite{FSZ}. We refer to~\cite{PandPixZvo} for an exposition of this CohFT that includes an overview of its constructions; the CohFT whose correlators give the ancestor potential at the point $\tau$ of this Frobenius manifold is called there the shifted Witten class of $A_n$ singularity. 
 
 \begin{theorem}\label{thm:an-singularity}
 	The correlation differentials of the global spectral curve data $\Sigma:=\mathbb{CP}^1$, $y:=p$ (the global coordinate), $x:=f(p,\tau)$, $B:=dp_1dp_2/(p_1-p_2)^2$ are expressed via Equation~\eqref{eq:CohFT-SpectralCurve-match} in terms of the shifted Witten class of $A_n$ singularity. 
 \end{theorem}

\begin{proof} As we have already mentioned, the function $\lambda=f(p,\tau)$ is known to be a superpotential of the corresponding Frobenius manifold. We have to show that this superpotential can be obtained by Dubrovin's construction in Section~\ref{sec:DubrovinsSuperpotential}. Then it is easy to see that all conditions of Theorem~\ref{thm:genus-0-identification} are satisfied, which implies this theorem. 

We construct solutions of Equation~\eqref{eq:Gauss-Manin-system-5-31-32} in terms of the integrals over the vanishing cycles.  Namely, consider the tangent bundle over the space of polynomials parametrized by $\tau\in\mathcal{T}$. It is identified with the space $\mathbb{C}[p]/(d_pf(p,\tau)/dp)$ by the map $v\mapsto (d_\tau f(p,\tau))(v)$ and equipped with a flat metric given by 
\begin{equation}
(v_1,v_2):=\res_{p=\infty} \frac{d_\tau f(v_1) d_\tau f(v_2)}{d_pf(p,\tau)/dp} dp.
\end{equation}

For a cycle $\beta\in H_0(f^{-1}(\lambda),\mathbb{C})$ we denote by $I_\beta(\lambda,\tau)$ the section of the tangent bundle specified by the following formula:
\begin{equation}
(I_\beta(\lambda,\tau),v) := \int_\beta d_\tau f(v) \cdot \frac{dp}{d_pf}.
\end{equation} 
In normalized canonical coordinates $I_\beta(\lambda,\tau)$ is represented by the vector $\phi^{\beta}(\lambda,\tau)$ with components given by
\beq
\phi^\beta_i(\lambda,\tau):=\left(I_\beta(\lambda,\tau),\frac{\Delta_i^{\frac 12}}{\sqrt{2}} \frac{\partial}{\partial u_i}\right)
\eeq
is a solution of Equation~\eqref{eq:Gauss-Manin-system-5-31-32} (see \cite{GiventalAn}). Let us discuss the singularities of this solution, depending on $\beta$. 

Consider the $\lambda$-plane as the image of the map $\lambda=f(p,\tau)$. Let $u_1,\dots,u_n$ be the critical values of $f(p,\tau)$. We can alway choose a system of cuts $L_i$, $i=1,\dots,n$, from $u_i$ to infinity such that the preimage $f^{-1}(\mathbb{C}\setminus \cup_{i=1}^n L_i)$ is a union of $n+1$ disks, $D_0,D_1,\dots,D_n$, glued along the boundary cuts in the following way:
\begin{itemize}
	\item[--] $D_0$ is glued to $D_i$ along the boundary that is a double cover of $L_i$; in particular, their common boundary contains the critical preimage of $u_i$;
	\item[--] All other lifts of the cut $L_i$ are just cuts inside $D_j$, $j\not=0,i$; the endpoints of these cuts are non-critical preimages of $u_i$. 
\end{itemize}
For $n=1,2,3$ we give the corresponding pictures for a real orientable blow-up at  infinity (that is, the boundary circle on the picture corresponds to the infinity point of the source sphere). The domain $D_0$ is shadowed there.

\begin{figure}[h]
\begin{subfigure}[t]{.3\textwidth}
  \begin{tikzpicture}
  [x=\picsize,y=\picsize,
  dott/.style={circle,draw=black,fill=black,inner sep=0pt,minimum size=3pt},
  dottblue/.style={circle,draw=blue,fill=blue,inner sep=0pt,minimum size=3pt}
  ]
  \useasboundingbox (-\rectsize,-\rectsize) rectangle (\rectsize,\rectsize);
  \draw[pattern=north east lines, pattern color=lightgray, dashed]
  (1,0) arc (0:180:1) -- cycle;
  \draw (0,0) circle (1);
  \node (circ1) [dott] at (0 : 1) {};
  \node (circ2) [dott] at (180 : 1) {};
  \node (center) [dottblue,label=below:$\textcolor{blue}{u_1}$] at (0,0) {};
  \draw [thick,blue] (circ1) -- (center) -- (circ2);
  \end{tikzpicture}
  \subcaption{$A_1$}
  \label{fig:sfig1}
\end{subfigure}
\begin{subfigure}[t]{.3\textwidth}
  \begin{tikzpicture}
  [x=\picsize,y=\picsize,
  dott/.style={circle,draw=black,fill=black,inner sep=0pt,minimum size=3pt},
  dottblue/.style={circle,draw=blue,fill=blue,inner sep=0pt,minimum size=3pt},
  dottgreen/.style={circle,draw=green!50!black,fill=green!50!black,inner sep=0pt,minimum size=3pt}
  ]
  \useasboundingbox (-\rectsize,-\rectsize) rectangle (\rectsize,\rectsize);
  \draw (0,0) circle (1);
  \node (circ1) [dott] at (30 : 1) {};
  \node (circ2) [dott] at (150 : 1) {};
  \node (circ3) [dott] at (270 : 1) {};
  \node (u1l) [dottblue,label=left:$\textcolor{blue}{u_1}$] at (-0.3,0) {};
  \node (u1r) [dottblue,label=below:$\textcolor{blue}{u_1}$] at (0.8,0) {};
  \node (u2l) [dottgreen,label=below:$\textcolor{green!50!black}{u_2}$] at (-0.8,0) {};
  \node (u2r) [dottgreen,label=right:$\textcolor{green!50!black}{u_2}$] at (0.3,0) {};
  \draw [thick,blue] (circ2) -- (u1l) -- (circ3);
  \draw [thick,green!50!black] (circ3) -- (u2r) -- (circ1);
  \draw [thick,green!50!black] (circ2) -- (u2l);
  \draw [thick,blue] (circ1) -- (u1r);
  \begin{scope}[on background layer]
  \draw[pattern=north east lines, pattern color=lightgray, dashed]
  (150:1) -- (-0.3,0) -- (270:1) -- (0.3,0) -- (30:1) arc (30:150:1);
  \end{scope}
  \end{tikzpicture}  
  \subcaption{$A_2$}
  \label{fig:sfig2}
\end{subfigure}
\begin{subfigure}[t]{.36\textwidth}
\begin{tikzpicture}
      [x=\picsize, y=\picsize,
      dott/.style={circle,draw=black,fill=black,inner sep=0pt,minimum size=3pt},
      dottblue/.style={circle,draw=blue,fill=blue,inner sep=0pt,minimum size=3pt},
      dottgreen/.style={circle,draw=green!50!black,fill=green!50!black,inner sep=0pt,minimum size=3pt},
      dottred/.style={circle,draw=red,fill=red,inner sep=0pt,minimum size=3pt}
      ]
      \useasboundingbox (-\rectsize,-\rectsize) rectangle (\rectsize,\rectsize);
      \draw (0,0) circle (1);
      
      \node (circ0) [dott] at (45 : 1) {};
      \node (circ90) [dott] at (135 : 1) {};
      \node (circ180) [dott] at (225 : 1) {};
      \node (circ270) [dott] at (315 : 1) {};
      
      \node (u1main) [dottblue,label=left:$\textcolor{blue}{u_1}$] at (150:0.3) {};
      \node (u2main) [dottgreen,label=below:$\textcolor{green!50!black}{u_2}$] at (270:0.3) {};
      \node (u3main) [dottred,label=right:$\textcolor{red}{u_3}$] at (30:0.3) {};
      
      \begin{scope}[shift={(circ90)}]
        \node (u3aux1) [dottred,label=below:$\textcolor{red}{u_3}$] at (250:0.5) {};
        \node (u2aux1) [dottgreen,label=below:$\textcolor{green!50!black}{u_2}$] at (275:0.5) {};  
      \end{scope}
      
      \begin{scope}[shift={(circ0)}]
        \node (u1aux3) [dottblue,label=below:$\textcolor{blue}{u_1}$] at (-70:0.5) {};
        \node (u2aux3) [dottgreen,label=below:$\textcolor{green!50!black}{u_2}$] at (-95:0.5) {};  
      \end{scope}
      
      \begin{scope}[shift={(circ180)}, rotate=10]
        \node (u3aux12) [dottred,label=below:$\textcolor{red}{u_3}$] at (0:0.5) {};
      \end{scope}
      
      \begin{scope}[shift={(circ270)}, rotate=170]
        \node (u1aux23) [dottblue,label=below:$\textcolor{blue}{u_1}$] at (0:0.5) {};
      \end{scope}
      
      \draw [thick,blue] (circ90) -- (u1main) -- (circ180);
      \draw [thick,green!50!black] (circ180) -- (u2main) -- (circ270);
      \draw [thick,red] (circ270) -- (u3main) -- (circ0);
      
      \draw [thick,green!50!black] (circ90) -- (u2aux1);
      \draw [thick,red] (circ90) -- (u3aux1);
      
      \draw [thick,blue] (circ0) -- (u1aux3);
      \draw [thick,green!50!black] (circ0) -- (u2aux3);
      
      \draw [thick,red] (circ180) -- (u3aux12);
      \draw [thick,blue] (circ270) -- (u1aux23);

      
      \begin{scope}[on background layer]
        \path [pattern=north east lines, pattern color=lightgray, dashed]
        (0,0) circle (1);
        \path [fill = white, opacity=1] (u3main) -- (circ270) arc (-45:45:1) -- (u3main);
        \path [fill = white, opacity=1] (u1main) -- (circ90) arc (135:225:1) -- (u1main);
        \path [fill = white, opacity=1] (u2main) -- (circ180) arc (225:315:1) -- (u2main);
      \end{scope}
      \end{tikzpicture}
  \subcaption{$A_3$}
  \label{fig:sfig3}
\end{subfigure}
\caption{$A_n$-singularities, $1\leq n \leq 3$}
\label{fig:fig}
\end{figure}
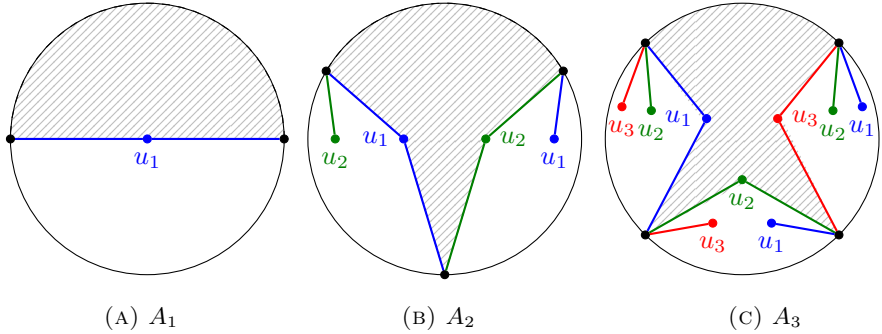

Consider the vanishing cycles $\beta_i\in H_0(f^{-1}(\lambda),\mathbb{C})$ given by $\beta_i:=p_0-p_i$, where $\lambda=f(p_0,\tau)=f(p_i,\tau)$, and $p_0\in D_0$, $p_i\in D_i$. Then the system of solutions of Equation~\eqref{eq:Gauss-Manin-system-5-31-32}
given by $\phi^{(i)}(\lambda,\tau):=\phi^{\beta_i}(\lambda,\tau)$ satisfies the properties given by Equations~\eqref{eq:DubrovinProp-1}--\eqref{eq:DubrovinProp-4}. In particular, $G^{ij}=1/2$ for $i\not=j$ and $G^{ii}=1$. The inverse matrix is given by $G_{ii}=2n/(n+1)$ and $G_{ij}=-2/(n+1)$ for $i\not=j$. Therefore, Dubrovin's solution $\phi=\sum_{i,j=1}^n G_{ij}\phi^{(j)}$ is equal to $\phi^{\beta_0}$ for 
\begin{equation}
\beta_0=\sum_{i=1}^n\left(\frac{2n}{n+1}-(n-1)\frac{2}{n+1}\right)\beta_i = 2p_0-\frac{2}{n+1}\sum_{i=0}^n p_i.
\end{equation}

Recall that for the Frobenius structure $A_n$, $d=(n-1)/(n+1)$. Also we recall that for the Euler vector field $E=\sum_{i=1}^n u_i \frac{\partial}{\partial u_i}$ and the unit vector field $e=\sum_{i=1}^n \frac{\partial}{\partial u_i}$ so that we have:
\begin{align}
Ef(p,\tau)& = f(p,\tau)-\frac{p}{n+1} \frac{d_pf(p,\tau)}{dp}, \\
ef(p,\tau)& = 1.
\end{align}
The formula $\phi^T (U-\lambda) \Psi \un$ can be written as 
$\left(I_{\beta_0}(\lambda,\tau),(E-\lambda e)/\sqrt{2}\right)$. Therefore,
\begin{align}
& \frac{\sqrt{2}}{1-d}\phi^T (U-\lambda) \Psi \un  = \frac{n+1}{2}\int_{\beta_0} (E-\lambda e) f(p,\tau) \cdot \frac{dp}{d_pf(p,\tau)} \\ \notag
& = \frac{n+1}{2} 
\int_{\beta_0} \left(f(p,\tau)-\frac{p}{n+1} \frac{d_pf(p,\tau)}{dp} -\lambda\right) \cdot \frac{dp}{d_pf(p,\tau)} .
\end{align}
Since the cycle $\beta_0$ lies in $f^{-1}(\lambda)$, then  $(f(p,\tau)-\lambda)|_{\beta_0}=0$. Therefore, the last integral can be rewritten as 
\begin{equation}
\frac{n+1}{2} 
\int_{\beta_0} \frac{p}{n+1} = p_0(\lambda,\tau)-\frac{1}{n+1} \sum_{i=0}^{n} p_i(\lambda,\tau).
\end{equation}
Since $\sum_{i=0}^{n} p_i(\lambda,\tau)=0$ (recall the form of the polynomial $f(p,\lambda)$), we conclude  that the function $\frac{\sqrt{2}}{1-d}\phi^T (U-\lambda) \Psi \un$ is equal to the branch $D_0$ of $p=f^{-1}(\lambda,\tau)$. 

So, $p(\mathbb{C}\setminus \cup_{i=1}^n L_i)=D_0$. It is obvious that $d_\lambda p$ has no zeros in $\mathbb{C}\setminus \cup_{i=1}^n L_i$, so $\hat D=D$, and one of the possible analytic continuation of the function $\lambda=f(p,\tau)|_{D_0}$ is its extension to the polynomial $f(p,\tau)$ defined on $\mathbb{CP}^1$. All condition of Theorem~\ref{thm:genus-0-identification} are satisfied, so we apply it here to complete the proof.
\end{proof}

\begin{remark}[Relation to Milanov's spectral curve]
	The global spectral curve that we constructed differs from the one constructed by Milanov in~\cite{MilanovGlobal}. Milanov gets a spectral curve with the same local behavior as $x=f(y,\tau)$ near the critical points, but, in our terms, he chooses a different analytic continuation of $\lambda|_D$. He constructs an analytic continuation using the action of the Weyl group (we revisit his construction in our terms in Section~\ref{sec:general-theory}), and obtains a curve where all preimages of the critical points in the $x$-plane are critical. In our terms, this can be achieved by gluing $n!$ copies of the curve $x=f(y,\tau)$ along the cuts connecting the non-critical preimages of the points $u_i$, $i=1,\dots,n$ such that each point belongs to exactly one cut. This makes all preimages of $u_1,\dots,u_n$ critical and will produce a curve of genus $1+\frac{n!}{2}(\frac{n^2}{2}-\frac n2 -2)$ where the function $x$ has degree $(n+1)!$, and it has $n!$ poles of degree $(n+1)$ each (cf. computation in~\cite{MilanovGlobal} and further explanation in Section~\ref{sec:general-theory}). 
\end{remark}

\subsection{Bouchard-Eynard recursion}

In this Section we discuss an application of Theorem~\ref{thm:an-singularity}. There is a more general formulation of topological recursion that works for functions $x$ with higher order singular points~\cite{BouchardEynard}. Locally, a higher order singularity is given by $x=y^{n+1}$, $B=dy_1dy_2/(y_1-y_2)^2$. Bouchard and Eynard announced a theorem~\cite{BouchardEynardThm} that identifies the coefficients of the local expansion in $y$ at $y=0$ of the correlation differentials of this spectral curve with the coefficients of the string solution of the $(r+1)$-Gelfand-Dickey hierarchy, also known as the total descendant potential of the $A_r$ singularity. The proof of Bouchard and Eynard goes through analysis of matrix models. Here we give a new proof of their theorem, namely, we derive it directly from Theorem~\ref{thm:an-singularity}. 

\begin{theorem} \cite{BouchardEynardThm} \label{th:Bouchard-Eynard}
	The Bouchard-Eynard recursion applied to $x=p^{n+1}$, $y=p$, $B=dp_1dp_2/(p_1-p_2)^2$ produces differentials $\omega_{g,k}$, whose expansions near infinity are given by 
	\begin{align}\label{eq:Bouchard-Eynard}
	& \omega_{g,k}(p_1,\dots,p_k) = \sum_{\substack{0\leq a_1,\dots,a_k \leq n-1\\ d_1,\dots,d_k}}
	\langle \tau_{d_1a_1}\cdots \tau_{d_ka_k} \rangle_{g,k} \times
	\\ \notag
	& 
	\prod_{j=1}^k
	\left(\frac{(a_j+1)(a_j+1+(n+1))\cdots (a_j+1+d_j(n+1))}{(-1)^{d_j} (n+1)^{d_j+1}} \frac{dp_j}{p_j^{(n+1)d_j+a_j+2}}\right),
	\end{align}
	where $\langle\tau_{d_1a_1}\cdots \tau_{d_ka_k} \rangle_{g,k}$ are the coefficients of the string solution of the $(n+1)$-Gelfand-Dickii hierarchy~\cite{WittenConj,GiventalAn,FSZ}.
\end{theorem}

Note that we don't recall and don't use the definition of the Bouchard-Eynard recursion. The only property that we are using here is that it is compatible with the usual recursion on the curves with simple singularities and the limits~\cite{BouchardEynard}. In this case, we know that in the neighborhood of infinity the correlation differentials of the Bouchard-Eynard recursion are the limits for $\epsilon\to 0$ of the correlation differentials of the usual recursion applied to $x=y^{n+1}+\epsilon y$. 

Let us now prove theorem~\ref{th:Bouchard-Eynard}.

\begin{proof}
	The flat coordinates $t_0=t_{\un},t_1,\dots,t_{n-1}$ are given on the space of polynomials $f(p,\tau)$ defined in Equation~\eqref{eq:polynomial} by the following formula:
\begin{equation}
f(p,\tau)^{\frac{1}{n+1}}=p+\frac{1}{n+1} \left( \frac{t_{n-1}}{p}+\frac{t_{n-2}}{p^2}+\cdots+\frac{t_0}{p^n} \right) + O\left(\frac{1}{p^{n+1}}\right). 
\end{equation}
Recall that the canonical coordinates are the critical values $u_1,\dots,u_n$ of $f(p,\tau)$ and it is obvious that $\partial u_i/\partial t_{\un}=1$. We denote by $c_1,\dots,c_n$ the positions of the critical points of function $f(p,\tau)$; so $u_i=f(c_i,\tau)$, $i=1,\dots,n$.  

We perform all computations only on a special curve in the space of polynomials, namely, $f(p,\tau)=p^{r+1}+\epsilon p$, and we are interested in all results only up $O(\epsilon)$ for $\epsilon\to 0$. In particular, we note that $t_a=O(\epsilon)$, $a=0,\dots,n-1$.

The full Jacobian of the change from the canonical to flat coordinates is then given by the following computation:
\begin{equation} \label{eq:full-jacobian}
\frac{\partial u_i}{\partial t_a} = \frac{\partial f(c_i,\tau)}{\partial t_a} =  f(c_i,\tau)^{\frac{n}{n+1}}\left(\frac{1}{c_i^{n-a}} + O(\epsilon)\right) = 
c_i^a \frac{\partial u_i}{\partial t_0} + O(\epsilon) = c_i^a +O(\epsilon). 
\end{equation}

The correlation differentials, written in terms of a CohFT considered in normalized canonical frame in Equation~\eqref{eq:CohFT-SpectralCurve-match}, can be rewritten in terms of the correlators of the ancestor potential of Givental~\cite{GiventalAn} $\cal A_t(\{t_{d,\alpha}\})$ considered at the point $t$ in flat coordinates $t_{d,\alpha}$, $d=0,1,2,\dots,$, $\alpha=0,\dots,n-1$ in the following way:
\begin{align}
\omega_{g,k}& =\sum_{\substack{i_1,\dots,i_k \\ d_1,\dots,d_k}} \int_{\overline{\mathcal{M}}_{g,k}} \alpha_{g,k}\left(\Delta_{i_1}^{\frac{1}{2}} \frac{\partial}{\partial u_{i_1}},\dots,\Delta_{i_k}^{\frac{1}{2}} \frac{\partial}{\partial u_{i_k}}\right) \prod_{j=1}^k
\psi_j^{d_j} d\left(\left(\frac{d}{dx}\right)^{d_j} \xi_{i_j}\right). 
\\ \notag
&=\sum_{\substack{\alpha_1,\dots,\alpha_k \\ d_1,\dots,d_k}}
\langle \tau_{d_1\alpha_1}\cdots \tau_{d_k\alpha_k} \rangle_{g,k} (t) \prod_{j=1}^k
d\left(\left(\frac{d}{dx}\right)^{d_j} \frac{-1}{p^{\alpha_j+1}}\right) +O(\epsilon). 
\end{align}
(by $\langle \tau_{d_1\alpha_1}\cdots \tau_{d_k\alpha_k} \rangle_{g,k} (t)$ we denote the coefficients of the expansion of $\log \mathcal{A}_t$).
Indeed, let us expand the vector $\sum_{i=1}^n  \xi_i(p) \Delta_{i}^{\frac{1}{2}} \frac{\partial}{\partial u_{i}}$ near $p=\infty$.  Recall that we denote by $c_1,\dots,c_n$ the positions of the critical points of function $f(p,\tau)$. We have:
\begin{align}
& \sum_{i=1}^n \xi_i(p)  \Delta_{i}^{\frac{1}{2}} \frac{\partial}{\partial u_{i}}  = \sum_{i=1}^n \left.\left(\frac{1}{z-p} \frac{dz}{d\sqrt{f(p,\tau)-u_i}} \right) \right|_{z=c_i} \Delta_{i}^{\frac{1}{2}} \frac{\partial}{\partial u_{i}} 
\\ \notag 
& = - \sum_{k=0}^\infty \frac{1}{p^{k+1}} \sum_{i=1}^n c_i^k \frac{\partial}{\partial u_{i}}  = -\sum_{k=0}^{n-1} \frac{1}{p^{k+1}} \frac{\partial}{\partial t_k} +O(\epsilon).
\end{align}
(we use Equation~\eqref{eq:full-jacobian} and the fact that $c_i^n=O(\epsilon)$ for the last equality).

Recall~\cite{PandPixZvo} that the correlators of the ancestor potential $\mathcal{A}_t$  are represented in terms of the correlators of the descendant potential (which is exactly the string solution of the $(n+1)$-Gelfand-Dickey hierarchy) as
\begin{equation}
\langle \tau_{d_1\alpha_1}\cdots \tau_{d_k\alpha_k} \rangle_{g,k} (t) = \langle \tau_{d_1\alpha_1}\cdots \tau_{d_k\alpha_k} \rangle_{g,k} +O(\epsilon). 
\end{equation}
Thus we see that 
\begin{align}
\omega_{g,k}& =\sum_{\substack{\alpha_1,\dots,\alpha_k \\ d_1,\dots,d_k}}
\langle \tau_{d_1\alpha_1}\cdots \tau_{d_k\alpha_k} \rangle_{g,k} \prod_{j=1}^k
d\left(\left(\frac{d}{dx}\right)^{d_j} \frac{-1}{p^{\alpha_j+1}}\right) +O(\epsilon). 
\end{align}
In the limit $\epsilon\to 0$ we get exactly Equation~\eqref{eq:Bouchard-Eynard}.
\end{proof}


\section{Frobenius manifolds for hypermaps}
\label{sec:hypermaps}

In this Section we construct a global spectral curve for the Frobenius manifold given by the superpotential $\lambda=f(p,a)$, where $a=(a_0,\dots,a_{n+1})$, $n\geq 1$, and
\begin{equation}\label{eq:hypermaps-superpotential}
f(p,a)=p^n+a_2 p^{n-2}+a_3 p^{n-3} +\cdots + a_n + \frac{a_{n+1}}{p-a_1}.
\end{equation}
This superpotential defines a semi-simple Frobenius manifold (this Frobenius manifold is studied in~\cite[Section 5]{Dub2dTFT}).  Furthermore the spectral curve $$(\Sigma,x,y,B)=\left(\mathbb{CP}^1,f(p,a),p,\frac{dp_1dp_2}{(p_1-p_2)^2}\right)$$ satisfies equations \eqref{eq:Identification-ConstantsDelta}-\eqref{eq:decomposition-doubleLaplace-B} hence it stores the correlators of a CohFT via Equation~\eqref{eq:CohFT-SpectralCurve-match}.  The following theorem answers the question of whether these two CohFTs coincide.
\begin{theorem}\label{thm:hypermaps} The CohFT associated to the Frobenius manifold given by the superpotential $\lambda=f(p,a)$ coincides with the one reconstructed from the spectral curve $(\mathbb{CP}^1,f(p,a),p,dp_1dp_2/(p_1-p_2)^2)$.
\end{theorem}

\begin{remark}
The correlation differentials for this spectral curve considered for the particular values of the parameters $a$ enumerate hypermaps on the curves. This is proved in~\cite{DOPS14}, see also~\cite{DoManescu}, where some special case of that was conjectured. 

So, Theorem~\ref{thm:hypermaps} is to be used in the converse way: We start with a combinatorial problem that is known to be solved by global topological recursion. It appears that the correlators of this global topological recursion are expressed in terms of a CohFT. This CohFT appears to be homogeneous, so it is associated to a Frobenius manifold, and this Theorem describes precisely the underlying Frobenius manifold. 
\end{remark}

\begin{proof} The proof is completely parallel to the proof of Theorem~\ref{thm:an-singularity}. Note that as in the case of $A_n$-singularity, we claim that the spectral curve is the superpotential itself.
	We use Theorem~\ref{thm:genus-0-identification}, so it is enough to show that we can reproduce the superpotential $\lambda=f(p,a)$ via Dubrovin's construction from Section~\ref{sec:DubrovinsSuperpotential}. 
	
	The canonical coordinates $u_1,\dots,u_{n+1}$ of this Frobenius manifold are the critical values of $f(p,a)$; the Euler vector field is given by 
	\begin{equation}
	E=\sum_{i=1}^n u_i \frac{\partial}{\partial u_i}=\sum_{i=1}^{n+1} \frac{i}{n} a_i \frac{\partial}{\partial a_i};	
	\end{equation}
 	the unit vector field is equal to 
 	\begin{equation}
 	e=\sum_{i=1}^n \frac{\partial}{\partial u_i}= \frac{\partial}{\partial a_n};
 	\end{equation}
	and the constant $d$ is equal to $(n-2)/n$. Note that 
	\begin{equation}\label{eq:Euler-Hypermaps}
	Ef(p,a) = f(p,a)-\frac{p}{n} \frac{df(p,a)}{dp}.
	\end{equation}

	As in the case of $A_n$ singularity, the solutions to the Equation~\eqref{eq:Gauss-Manin-system-5-31-32} are given by the integrals over the cycles $\beta\in H_0(f^{-1}(\lambda),\mathbb{C})$, where the components of the solutions are given by 
	\begin{equation}
	\phi^\beta_i:=\int_\beta \frac{\Delta_i^{\frac 12}}{\sqrt{2}} \frac{\partial f(p,a)}{\partial u_i} \cdot \left(\frac{df(p,a)}{dp}\right)^{-1}.
	\end{equation}
	
	Consider the $\lambda$-plane as the image of the map $\lambda=f(p,\tau)$. Recall that $u_1,\dots,u_{n+1}$ are the critical values of $f(p,\tau)$. We can alway choose a system of cuts $L_i$, $i=1,\dots,n+1$, from $u_i$ to infinity such that the preimage $f^{-1}(\mathbb{C}\setminus \cup_{i=1}^n L_i)$ is a union of $n+1$ disks, $D_0,D_1,\dots,D_n$, glued along the boundary cuts in the following way:
	\begin{itemize}
		\item[--] $D_0$ is glued to $D_i$, $i=2,\dots,n$, along the boundary that is a double cover of $L_i$; in particular, their common boundary contains the critical preimage of $u_i$;
		\item[--] $D_0$ is glued to $D_1$ along two components of the boundary that are double covers of $L_1$ and $L_{n+1}$ and these boundary components have common point $p=a_{1}$. In particular, these boundary components contain the critical preimages of $u_1$ and $u_{n+1}$;
		\item[--] All other lifts of the cut $L_i$ are just cuts inside $D_j$, $j\not=0,i$ for $i=2,\dots,n$ and $j\not=0,1$ for $i=1,n+1$; the endpoints of these cuts are non-critical preimages of $u_i$. 
	\end{itemize}
	For $n=1,2,3$ we give the corresponding pictures for a real orientable blow-up at  infinity (that is, the external boundary circle on the picture corresponds to the infinity point of the source sphere, and the internal circle corresponds to $p=a_1$). The domain $D_0$ is shadowed.
	
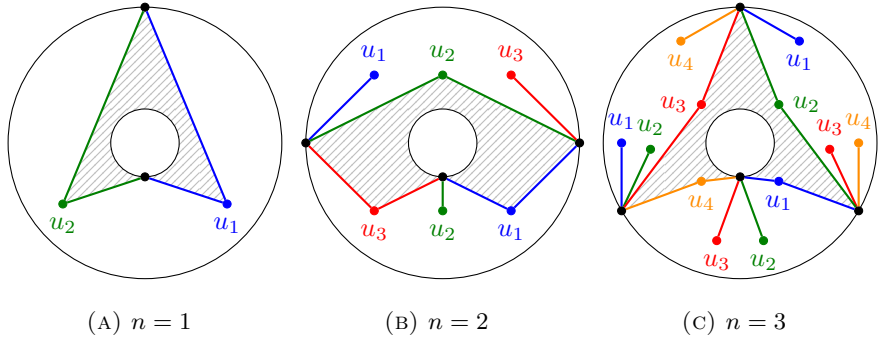
\begin{figure}[h]
\begin{subfigure}[t]{.3\textwidth}
  \begin{tikzpicture}
  [x=\picsize,y=\picsize,
  dott/.style={circle,draw=black,fill=black,inner sep=0pt,minimum size=3pt},
  dottblue/.style={circle,draw=blue,fill=blue,inner sep=0pt,minimum size=3pt},
  dottgreen/.style={circle,draw=green!50!black,fill=green!50!black,inner sep=0pt,minimum size=3pt}
  ]
    \useasboundingbox (-\rectsize,-\rectsize) rectangle (\rectsize,\rectsize);
  \draw (0,0) circle (1);
  \draw (0,0) circle (0.25);
  \node (circlarge1) [dott] at (90 : 1) {};
  \node (circsmall1) [dott] at (270 : 0.25) {};
  \node (u1) [dottblue,label=below:$\textcolor{blue}{u_1}$] at (0.6,-0.45) {};
  \node (u2) [dottgreen,label=below:$\textcolor{green!50!black}{u_2}$] at (-0.6,-0.45) {};
  \draw [thick,blue] (circlarge1) -- (u1) -- (circsmall1);
  \draw [thick,green!50!black] (circlarge1) -- (u2) -- (circsmall1);
  \begin{scope}[on background layer]
  \fill[pattern=north east lines, pattern color=lightgray, dashed]
  (90:1) -- (-0.6,-0.45) -- (270:0.25) arc (270:-90:0.25) -- (0.6,-0.45);
  \end{scope}
  \end{tikzpicture}
  \subcaption{$n=1$}
  \label{fig:sfigg1}
\end{subfigure}
\begin{subfigure}[t]{.3\textwidth}
  \begin{tikzpicture}
  [x=\picsize,y=\picsize,
  dott/.style={circle,draw=black,fill=black,inner sep=0pt,minimum size=3pt},
  dottblue/.style={circle,draw=blue,fill=blue,inner sep=0pt,minimum size=3pt},
  dottgreen/.style={circle,draw=green!50!black,fill=green!50!black,inner sep=0pt,minimum size=3pt},
  dottred/.style={circle,draw=red,fill=red,inner sep=0pt,minimum size=3pt}
  ]
      \useasboundingbox (-\rectsize,-\rectsize) rectangle (\rectsize,\rectsize);
  \draw (0,0) circle (1);
  \draw (0,0) circle (0.25);
  \node (circlarge1) [dott] at (0 : 1) {};
  \node (circlarge2) [dott] at (180 : 1) {};
  \node (circsmall1) [dott] at (270 : 0.25) {};
  \node (u1s) [dottblue,label=above:$\textcolor{blue}{u_1}$] at (-0.5,0.5) {};
  \node (u1d) [dottblue,label=below:$\textcolor{blue}{u_1}$] at (0.5,-0.5) {};
  \node (u2s) [dottgreen,label=below:$\textcolor{green!50!black}{u_2}$] at (0,-0.5) {};
  \node (u2d) [dottgreen,label=above:$\textcolor{green!50!black}{u_2}$] at (0,0.5) {};
  \node (u3s) [dottred,label=above:$\textcolor{red}{u_3}$] at (0.5,0.5) {};
  \node (u3d) [dottred,label=below:$\textcolor{red}{u_3}$] at (-0.5,-0.5) {};
  \draw [thick,blue] (circlarge2) -- (u1s);
  \draw [thick,blue] (circlarge1) -- (u1d) -- (circsmall1);
  \draw [thick,green!50!black] (circsmall1) -- (u2s);
  \draw [thick,green!50!black] (circlarge1) -- (u2d) -- (circlarge2);
  \draw [thick,red] (circlarge1) -- (u3s);
  \draw [thick,red] (circlarge2) -- (u3d) -- (circsmall1);
  \begin{scope}[on background layer]
  \fill[pattern=north east lines, pattern color=lightgray, dashed]
  (270:0.25) -- (-0.5,-0.5) -- (180:1) -- (0,0.5) -- (0:1) -- (0.5,-0.5) -- (-90:0.25) arc (-90:270:0.25);
  \end{scope}
  \end{tikzpicture}
  \subcaption{$n=2$}
  \label{fig:sfigg2}
\end{subfigure}
\begin{subfigure}[t]{.3\textwidth}

\begin{tikzpicture}
[x=\picsize, y=\picsize,
dott/.style={circle,draw=black,fill=black,inner sep=0pt,minimum size=3pt},
dottblue/.style={circle,draw=blue,fill=blue,inner sep=0pt,minimum size=3pt},
dottgreen/.style={circle,draw=green!50!black,fill=green!50!black,inner sep=0pt,minimum size=3pt},
dottred/.style={circle,draw=red,fill=red,inner sep=0pt,minimum size=3pt},
dottdblue/.style={circle,draw=orange!90!yellow,fill=orange!90!yellow,inner sep=0pt,minimum size=3pt}
]
      \useasboundingbox (-\rectsize,-\rectsize) rectangle (\rectsize,\rectsize);
\draw (0,0) circle (1);
\draw (0,0) circle (0.25);

\node (circ1) [dott] at (90 : 1) {};
\node (circ2) [dott] at (210 : 1) {};
\node (circ3) [dott] at (330 : 1) {};
\node (circinner) [dott] at (270 : 0.25) {};

\node (u1main) [dottblue,label=below:$\textcolor{blue}{u_1}$] at (315:0.4) {};
\node (u2main) [dottgreen,label=right:$\textcolor{green!50!black}{u_2}$] at (45:0.4) {};
\node (u3main) [dottred,label=left:$\textcolor{red}{u_3}$] at (135:0.4) {};
\node (u4main) [dottdblue,label=below:$\textcolor{orange!90!yellow}{u_4}$] at (225:0.4) {};

\begin{scope}[shift={(circ1)}]
  \node (u1aux1) [dottblue,label=below:$\textcolor{blue}{u_1}$] at (-30:0.5) {};
  \node (u4aux1) [dottdblue,label=below:$\textcolor{orange!90!yellow}{u_4}$] at (210:0.5) {};
  \draw [thick,blue] (circ1) -- (u1aux1);
  \draw [thick,orange!90!yellow] (circ1) -- (u4aux1);
\end{scope}

\begin{scope}[shift={(circ2)}]
  \node (u1aux2) [dottblue,label=above:$\textcolor{blue}{u_1}$] at (90:0.5) {};
  \node (u2aux2) [dottgreen,label=above:$\textcolor{green!50!black}{u_2}$] at (65:0.5) {};  
  \draw [thick,blue] (circ2) -- (u1aux2);
  \draw [thick,green!50!black] (circ2) -- (u2aux2);
\end{scope}

\begin{scope}[shift={(circ3)}]
  \node (u3aux3) [dottred,label=above:$\textcolor{red}{u_3}$] at (115:0.5) {};
  \node (u4aux3) [dottdblue,label=above:$\textcolor{orange!90!yellow}{u_4}$] at (90:0.5) {};  
  \draw [thick,red] (circ3) -- (u3aux3);
  \draw [thick,orange!90!yellow] (circ3) -- (u4aux3);
\end{scope}

\begin{scope}[shift={(circinner)}, rotate=180]
  \node (u2auxinner) [dottgreen,label=below:$\textcolor{green!50!black}{u_2}$] at (110:0.5) {};
  \node (u3auxinner) [dottred,label=below:$\textcolor{red}{u_3}$] at (70:0.5) {};  
  \draw [thick,green!50!black] (circinner) -- (u2auxinner);
  \draw [thick,red] (circinner) -- (u3auxinner);
\end{scope}

\draw [thick,blue] (circinner) -- (u1main) -- (circ3);
\draw [thick,green!50!black] (circ3) -- (u2main) -- (circ1);
\draw [thick,red] (circ1) -- (u3main) -- (circ2);
\draw [thick,orange!90!yellow] (circ2) -- (u4main) -- (circinner);

\begin{scope}[on background layer]
  \path [pattern=north east lines, pattern color=lightgray, dashed]
  (0,0) circle (1);
  \path [fill = white, opacity=1] (0,0) circle (0.25);
  \path [fill = white, opacity=1] (u2main) -- (circ3) arc (-45:103:1) -- (u2main);
  \path [fill = white, opacity=1] (u3main) -- (circ1) arc (90:210:1) -- (u3main);
  \path [fill = white, opacity=1]
  (circinner) -- (u4main) -- (circ2) arc (210:330:1) -- (circinner);
  \begin{scope}[shift={(u1main)}]
    \path [fill = white, opacity=1] (0,0) arc (140:360:0.5);
    \path [fill = white, opacity=1] (0,0) arc (60:360:0.25);
  \end{scope}
  \begin{scope}[shift={(u4main)}, rotate=-85]
    \path [fill = white, opacity=1] (0,0) rectangle (0.1,0.25);
  \end{scope}
  \begin{scope}[shift={(u4main)}, rotate=-160]
    \path [fill = white, opacity=1] (0,0) rectangle (0.5,0.1);
  \end{scope}
\end{scope}
\end{tikzpicture}
  \subcaption{$n=3$}
  \label{fig:sfigg3}
\end{subfigure}
\caption{Hypermaps $z^n+\dfrac{1}{z}$, $1\leq n \leq 3$}
\label{fig:figg}
\end{figure} 

Consider the vanishing cycles $\beta_i\in H_0(f^{-1}(\lambda),\mathbb{C})$ given by $\beta_i:=p_i-p_0$, where $\lambda=f(p_0,\tau)=f(p_i,\tau)$, and $p_0\in D_0$, $p_i\in D_i$. Then the system of solutions of Equation~\eqref{eq:Gauss-Manin-system-5-31-32}
given by $\phi^{(i)}(\lambda,\tau):=\phi^{\beta_i}(\lambda,\tau)$, $i=1,\dots,n$, $\phi^{(n+1)}=\phi^{(1)}$, satisfies the properties given by Equations~\eqref{eq:DubrovinProp-1}--\eqref{eq:DubrovinProp-4}. In particular, $G^{ii}=1$ for $i=1,\dots,n+1$, $G^{1,n+1}=G^{n+1,1}=1$, and for all other $i\not=j$ $G^{ij}=1/2$. So, this matrix is degenerate.

However, Remark~\ref{rem:more-general-phi} specifies the properties of $\phi$ that are sufficient for Theorems~\ref{thm:right-function-y} and~\ref{thm:genus-0-identification}. Note that
\begin{equation}
\phi:=\frac{2}{n+1}\sum_{i=1}^{n} \phi^{(i)}=\phi^{\beta_0}
\end{equation}
for $\beta_0:=\frac{2}{n+1}\sum_{i=0}^{n} p_i - 2p_0$ satisfies all condition of Remark~\ref{rem:more-general-phi}. With this choice of $\phi$ and, therefore, $\beta_0$, Dubrovin's superpotential can be presented as
\begin{align}
& \frac{\sqrt{2}}{1-d}\phi^T (U-\lambda) \Psi \un  = \frac{n}{2}\int_{\beta_0} (E-\lambda e) f(p,a) \cdot \left(\frac{df(p,a)}{dp}\right)^{-1} .
\end{align} 
Using Equation~\eqref{eq:Euler-Hypermaps}, we have:
\begin{align}
& \frac{\sqrt{2}}{1-d}\phi^T (U-\lambda) \Psi \un = \frac{n}{2} 
\int_{\beta_0} \left(f(p,a)-\frac{p}{n} \frac{df(p,a)}{dp} -\lambda\right) \cdot \left(\frac{df(p,a)}{dp}\right)^{-1} \\ \notag
& = \int_{\beta_0} -\frac{p}2 = p_0-\frac{1}{n+1}\sum_{i=0}^n p_i =
\begin{cases} 
p_0-\frac{\lambda+a_1}{2} & n=1; \\
p_0-\frac{a_1}{n+1} & n>1.
\end{cases}
\end{align}
(in the last equality we used that we know the sum of all roots of the equation $f(p,a)=\lambda$).

Let us now discuss the cases that we get. For $n>1$ Dubrovin's function $\pdub=\pdub(\lambda,a)$ is the branch $D_0$ of the inverse function of $\lambda=f(p,a)$ shifted by a constant. Obviously, $d_\lambda \pdub$ has no zeros in $\mathbb{C}\setminus\cup_{i=1}^{n+1} L_i$, and we can choose as the analytic extension of $\lambda|_{D_0}$ the function $\lambda=f(p,a)$ defined on $\mathbb{CP}^1$. Then Theorem~\ref{thm:genus-0-identification} is applied. We get, therefore, not precisely the statement that we want to prove, but we have instead $y=\pdub=p-a_1/(n+1)$ (the Bergman kernel is still the same). However, it doesn't change anything in topological recursion if we shift $y$ by a constant. 

The case $n=1$ is even more interesting. One can easily check by direct computation that Dubrovin's function $\pdub=\pdub(\lambda,a)$ is equal to $\sqrt{(\lambda-u_1)(\lambda-u_2)}/2$. Further construction of the curve gives the following equation:
$$
\pdub^2-\frac 14 (\lambda-a_1)^2+a_{2}=0
$$
It is a rational curve, and it has a global coordinate $p=\pdub+(\lambda+a_1)/2$, which is our original coordinate $p$, that is, $\lambda=p+a_2/(p-a_1)$. Theorem~\ref{thm:genus-0-identification} can not be  applied directly, but in this case we can just check by hand that we get the statement that we want to prove.  See Appendix~\ref{subsec:discretesurface}.

Note that Theorem~\ref{thm:right-function-y} suggests that the right choice of function $y$ is $y=\pdub=p-(\lambda+a_1)/2$ rather than $y=p$. However, it doesn't change anything in topological recursion if we shift $y$ by a function of $x=\lambda$, so there is no contradiction.
\end{proof}


Since in this example we rather start from a combinatorial problem of enumeration of hypermaps and use Theorem~\ref{thm:hypermaps} in order to clarify the structure of the ELSV-type formula~\eqref{eq:CohFT-SpectralCurve-match} for this combinatorial problem, it is interesting to have a description of the underlying Frobenius manifolds (given by superpotentials) in terms of their prepotentials. We know an algorithm which can produce this prepotential for any given $n$ (this algorithm follows from Dubrovin's construction found in \cite{Dub2dTFT}), but we do not know a general formula which would describe these prepotentials for all $n\geq 1$. Here we list the formulas for cases $n=1,2,3$:

\begin{align*}
&n=1: &  & \frac{a_1^2 a_2}{2} +\frac{a_2^2}{2}\log a_2; \\ 
&n=2: &  & \frac{a_1^3}{6} +a_1a_2a_3+\frac{a_{3}^2}{2}\log a_{3}+\frac{a_1^3a_3}{6} -\frac{3a_3^2}{4}; \\
&n=3: &  & \frac{a_1^2 a_4}{2} + a_1 a_2 a_3 - \frac{3 a_2^2}{4}+\frac{a_2^2}{2}\log \left(a_2\right) + \frac{a_2 a_3^4}{4} +\frac{3 a_2 a_3^2 a_4}{2} +\frac{3 a_2 a_4^2}{2} -\frac{3 a_4^4}{8}
\end{align*}

Note that in the case $n=1$ the corresponding combinatorial problem has also interpretation in terms of the discrete volumes of the moduli space of curves~\cite{NorburyLattice} and discrete surfaces/generalized Catalan numbers~\cite{AlexMirMor,DMSS12,EO09}. The relation of these combinatorial problems to a CohFT is also discussed in~\cite{ACNP,FangLiuZong}, though it is not mentioned there that the underlying Frobenius manifold is given by the prepotential $a_1^2a_2/2+a_2^2/2\cdot \log a_2$.

\section{Elliptic example}
\label{sec:elliptic}
In this section we give an example of a superpotential that satisfies the conditions of Theorem~\ref{thm:any-genus-identification}.

Consider the spectral curve defined by the Weierstrass $\wp$-function
\beq  \label{specweier}
\lambda=\wp(z),\quad p=z,\quad B(z,z')=\left(\wp(z-z')+b\right)dzdz'
\eeq
where $b\in\C$ and $p$ is only defined locally---it is the primitive of a holomorphic differential on the curve---which is sufficient for topological recursion.   The compatibility condition \eqref{dosstest} is satisfied by Proposition~\ref{th:yprim}.  It is equivalent to the elliptic identity:
\begin{equation}  \label{ellipticDT}
\frac{\wp''(z)}{\wp'(z)^2}=\sum_{i=1}^3\frac{\wp(z-\omega_i)}{\wp''(\omega_i)}
\end{equation}
where the sum is over the zeros $\omega_i$ of $\wp'(z)$.  Hence the spectral curve defines a CohFT. 

Introduce three parameters $\omega$, $\omega'$ and $c$ into the spectral curve to define the following superpotential taken from \cite{Dub2dTFT}:
\begin{equation}  \label{dubsup}
\lambda=\wp(z;\omega,\omega')+c,\quad p=\frac{z}{\omega}
\end{equation}
where
\begin{equation}  \label{wp}
\wp(z;\omega,\omega')=\frac{1}{z^2}+\sum_{(m,n)\neq (0,0)}\frac{1}{(z-2m\omega-2n\omega')^2}-\frac{1}{(2m\omega+2n\omega')^2}.
\end{equation}

The Frobenius manifold structure on $M=\{(\omega,\omega',c)\}$ is given by the formulae \eqref{spmet}, \eqref{spint}, \eqref{sp3pt} where the vector fields $\partial$ on $M$ are given by, for example $\partial_\omega$, $\partial_{\omega'}$, $\partial_c$.   Note that in \eqref{ellipticDT}, $\omega_1=\omega$, $\omega_2=\omega'$, $\omega_3=\omega+\omega'$.  

\begin{remark}
Note that we know that the superpotential \eqref{dubsup} defines a Frobenius manifold due to the existence of flat coordinates, proven in \cite{Dub2dTFT}, and also given below.  The CohFT produced by topological recursion applied to \eqref{specweier} is homogeneous if we choose $b$ in \eqref{specweier} so that $\int_AB(z,z')=0$, i.e. $b=\eta/\omega$ where the $A$ and $B$ periods are $2\omega=\oint_Adz$, $2\omega'=\oint_Bdz$ and
$$\eta=-\frac{1}{2}\oint_A\wp(z)dz,\quad\eta'=-\frac{1}{2}\oint_B\wp(z)dz.
$$
The periods satisfy Legendre's relation
$$ \eta\omega'-\eta'\omega=\frac{i\pi}{2}.$$
The homogeneous CohFT corresponds to a conformal Frobenius manifold which gives rise to a superpotential via Dubrovin's construction (actually, since $d=1$ it is a variant of the construction).  What needs to be proven is that the two superpotentials agree.  \end{remark}
 
\begin{theorem}  \label{th:elliptic}
The superpotential \eqref{dubsup} can be obtained via (a variant of) Dubrovin's construction described in Section~\ref{sec:DubrovinsSuperpotential} applied to the Frobenius manifold $M$.  The conditions of Theorem~\ref{thm:any-genus-identification} are satisfied for this superpotential.  Hence the two cohomological field theories---obtained from the superpotential \eqref{dubsup} and topological recursion applied to the spectral curve \eqref{specweier} with $b=\eta/\omega$---agree.
\end{theorem}
\begin{proof}
To apply Dubrovin's construction to $M$ we construct a solution of the Gauss-Manin system as in the proof of Theorem~
\ref{thm:an-singularity}.

The flat metric \eqref{spmet} for the superpotential \eqref{dubsup} is given by  
\begin{equation}
(\partial,\partial'):=\sum_{i=1}^3\res_{z=\omega_i} \frac{\partial\lambda\cdot\partial'\lambda}{\wp'} dp.
\end{equation}
We use this to construct a vector field $I_\beta(\lambda;u)$ on $M$ for any cycle $\beta\in H_0(\lambda^{-1}({\rm pt}),\mathbb{C})$ specified by:
\begin{equation}
(I_\beta(\lambda;u),\partial) := \int_\beta\frac{\partial(\lambda)}{ d_p\lambda/dp}.
\end{equation} 
The elliptic curve \eqref{dubsup} is built by gluing two copies of the disk $D=\mathbb{C}\setminus \cup_{i=1}^3 L_i$ in the $\lambda$ plane along $L_i$.  Choose $\beta$ to be the cycle given by $p_0-p_1$, for $p_0$ and $p_1$ the pre-images of $\lambda$ in each of the two disks.  In normalized canonical coordinates $I_\beta(\lambda;u)$ is represented by a solution $\phi^{\beta}(\lambda;u)=\sum_i\phi^\beta_i(\lambda;u)\partial_{v_i}$ of the Gauss-Manin system \eqref{eq:Gauss-Manin-system-5-31-32}  which has components given by
\begin{align*}
\phi^\beta_i(\lambda;u)&=\left(I_\beta(\lambda;u),\frac{1}{\sqrt{2}}\Delta_i^{\frac 12} \partial_{u_i}\right)\\
&=\int_\beta\frac{\Delta_i^{\frac 12}\partial_{u_i}\lambda}{\sqrt{2}\cdot d_p\lambda/dp}\\
&=\frac{\sqrt{2}\cdot\Delta_i^{\frac 12}\partial_{u_i}\lambda}{d_p\lambda/dp}
\end{align*}
where the integral over $\beta$ simply doubles the integrand since the integrand is skew symmetric.  Since $d=1$, we cannot use the inversion formula \eqref{gradtocoord} so we directly check that $\frac{1}{\sqrt{2}}\phi^\beta=\nabla_u p$ as follows.
$$\nabla_up=\eta^{-1}d_up=\frac{1}{\omega}\sum_{i=1}^3\Delta_i\cdot\frac{\partial_{u_i}\lambda}{\wp'(z)}\partial_{u_i}
=\frac{1}{\omega}\sum_{i=1}^3\frac{\Delta_i^{\frac12}\partial_{u_i}\lambda}{\wp'(z)}\partial_{v_i}=\frac{1}{\sqrt{2}}\phi^\beta.$$

Using
$$\partial_{u_i}\lambda=\frac{1}{2\wp''(\omega_i)}\frac{\wp'(z)^2}{\wp(z)-\wp(\omega_i)}+\frac{z\wp(\omega_i)+\zeta(z)}{\wp''(\omega_i)}\wp'(z)$$
which can be proven from the known variations $\sum_iu_i^k\partial_{u_i}$, $k=0,1,2$, we see that the solution $\phi^{\beta}(\lambda;u)$ satisfies 
\begin{align*}
\frac{1}{\sqrt{2}}\phi^T  \Psi \un &=\frac{1}{\omega}\sum_{i=1}^3\frac{\partial_{u_i}\lambda}{\wp'(z)}\\
&=\frac{1}{\omega}\sum_{i=1}^3\frac{1}{2\wp''(\omega_i)}\frac{\wp'(z)}{\wp(z)-\wp(\omega_i)}+\frac{z\wp(\omega_i)+\zeta(z)}{\wp''(\omega_i)}\\
&=\frac{1}{\omega\wp'(z)}=\frac{dp}{d\lambda}
\end{align*}
which is \eqref{eq:dpdlambda} and hence via Remark~\ref{rem:more-general-phi} we see that the properties of $\phi$ are sufficient for Theorems~\ref{thm:right-function-y} and~\ref{thm:genus-0-identification}.
Hence the theorem follows.
\end{proof}

Theorem~\ref{th:elliptic} states that we can study the CohFT obtained from the superpotential \eqref{dubsup} via topological recursion applied to the spectral curve \eqref{specweier}.  We will need the three-point function of this CohFT in calculations below.  We calculate it in two ways to demonstrate the proof, although we know from the theorem that they coincide.

\subsection{Three-point function.}

{\em Superpotential.} Introduce the canonical coordinates
$$u_i=\wp(\omega_i)+c,\quad i=1,2,3
$$
where, as usual, $\omega_1=\omega$, $\omega_2=\omega'$, $\omega_3=\omega+\omega'$.  The three-point calculations take place in the ring $\C[E]/\wp'=\C[\wp]/\wp'$ and we have
$$\frac{\partial\lambda}{\partial u_1}\equiv\frac{(\lambda-u_2)(\lambda-u_3)}{(u_1-u_2)(u_1-u_3)},\quad \frac{\partial\lambda}{\partial u_1}\frac{\partial\lambda}{\partial u_j}\equiv \delta_{1j}\frac{\partial\lambda}{\partial u_1},\ j=2,3.
$$
and cyclic permutations of the above.
This is quite general and also can be proven via elliptic identities.  Hence the three-point function for the superpotential is
$$\left\langle\frac{\partial}{\partial u_i},\frac{\partial}{\partial u_i},\frac{\partial}{\partial u_i}\right\rangle=\left\langle\frac{\partial}{\partial u_i},\frac{\partial}{\partial u_i}\right\rangle=\left\langle\frac{\partial}{\partial u_i}\right\rangle=\sum_{j=1}^3\res_{z=\omega_j}\frac{\frac{\partial\lambda}{\partial u_i}}{\wp'(z)}\frac{dz}{\omega^2}=\frac{1}{\omega^2\wp''(\omega_i)}.$$
Thus
\beq  \label{3pt}
\left\langle\frac{\partial}{\partial v_i},\frac{\partial}{\partial v_i},\frac{\partial}{\partial v_i}\right\rangle=\omega\sqrt{\wp''(\omega_i)}
\eeq
where $\frac{\partial}{\partial v_i}=\Delta_i^{\frac{1}{2}}\frac{\partial}{\partial u_i}$ for $\Delta_i^{-1}=\left\langle\frac{\partial}{\partial u_i},\frac{\partial}{\partial u_i}\right\rangle=\omega^2\cdot\wp''(\omega_i)$ give the normalized canonical coordinates.

{\em Topological recursion.}  The three-point function obtained via topological recursion is
\begin{align*}
\omega_{0,3}(z_1,z_2,z_3)&=\sum_{j=1}^3\res_{z=\omega_j}\frac{\omega dz}{\wp'(z)}\prod_{i=1}^3(\wp(z_i-z)+b)dz_i\\
&=\sum_{j=1}^3\frac{\omega}{\wp''(\omega_j)}\prod_{i=1}^3(\wp(z_i-\omega_j)+b)dz_i\\
&=\sum_{j=1}^3\omega\sqrt{\wp''(\omega_j)}V^j_0(z_1)V^j_0(z_2)V^j_0(z_3)
\end{align*}
for 
$$V^j_0(z)=\left.\frac{(\wp(z-\omega_j)+b)dzdz_j}{ds_j}\right|_{s_j=0}=
\frac{\wp(z-\omega_j)+b}{\sqrt{\wp''(\omega_i)}}dz
$$
where $\lambda=\wp(z)+c=\frac{1}{2}s_j^2+\wp(\omega_j)+c$ defines the local coordinate $s_j$.  The coefficients of $V^j_0(z_i)$ define the three-point function of the cohomological field theory which agree with \eqref{3pt}.

\subsection{Flat coordinates}
The cohomological field theory is defined on the three-dimensional vector space $\C[\wp]/\wp'$ equipped with its natural ring structure and gives rise to a Frobenius manifold structure on the family $M$ of such rings parametrized by $\{\omega,\omega',c\}$.  It will be convenient to express the metric on $M$ with respect to a natural basis of vector fields on $M$ corresponding to the basis $\{1,\wp,\wp^2\}$ of $\C[\wp]/\wp'$ since the metric requires knowledge of the variation of $\wp$ under the action of vector fields on the Frobenius manifold.  We will see that $\{\omega,\omega',c\}$ are not flat coordinates and find in Lemma~\ref{th:flat} flat coordinates $\{t_1,t_2,t_3\}$ on $M$, i.e. so that the metric on $M$ is constant with respect to them. 

Recall that correlation functions of the cohomological field theory arising from topological recursion applied to a spectral curve appear as coefficents of auxiliary differentials on the spectral curve.  Proposition~\ref{th:auxflat} gives the auxiliary differentials that correspond to the flat basis for the metric. 

In the following lemma we calculate the vector fields on $M$ that correspond to the basis elements $1,\wp,\wp^2$ of $\C[\wp]/\wp'$.  This uses $g_2=g_2(\omega,\omega')$ defined by $\wp'(z)^2=4\wp(z)-g_2\wp-g_3$. 
\begin{lemma}
Under the map $TM\to\C[\wp]/\wp'$ defined by $\partial\mapsto\partial\lambda(\text{mod\ }\wp')$ for $\lambda=\wp(z;\omega,\omega')+c$
\beq\label{basis}
\partial_c\mapsto 1,\quad -\tfrac{1}{2}\left(\omega\partial_\omega+\omega'\partial_{\omega'}\right)\mapsto\wp,\quad -\tfrac{1}{2}\left(\eta\partial_\omega+\eta'\partial_{\omega'}\right)+\tfrac{1}{6}g_2\partial_c\mapsto\wp^2.
\eeq
\end{lemma}
\begin{proof}
The variation $\partial_c\lambda=1$ is obvious.  The identity
\begin{equation}
\omega\partial_\omega\wp(z)+\omega'\partial_{\omega'}\wp(z)+z\wp'(z)=-2\wp(z)
\end{equation}
follows immediately from the expansion \eqref{wp} of $\wp$ and yields $-\tfrac{1}{2}\left(\omega\partial_\omega+\omega'\partial_{\omega'}\right)\mapsto\wp$.  The final identification uses the identity proven in \cite{FS1882}
\begin{equation}
\eta\partial_\omega\wp(z)+\eta'\partial_{\omega'}\wp(z)+\zeta(z)\wp'(z)=-2\wp(z)^2+\tfrac{1}{3}g_2
\end{equation}
where $\zeta(z)$ is the Weierstrass $\zeta$-function 
$$\zeta(z;\omega,\omega')=\frac{1}{z}+\sum_{(m,n)\neq (0,0)}\frac{1}{z-2m\omega-2n\omega'}+\frac{1}{2m\omega+2n\omega'}+\frac{z}{(2m\omega+2n\omega')^2}$$
which is not an elliptic function (C.63 in \cite{Dub2dTFT}). 
Note that $\eta=\zeta(\omega)$, $\eta'=\zeta(\omega')$. 
\end{proof}
The metric 
$$\langle\wp^j,\wp^k\rangle=\sum_{i=1}^3\res_{z=\omega_i}\frac{\wp^{j+k}}{\wp'(z)}\frac{dz}{\omega^2}=-\res_{z=0}\frac{\wp^{j+k}}{\wp'(z)}\frac{dz}{\omega^2}
$$
is given by 
\begin{center}
  \begin{tabular}{c || c | c | c }
      & \quad\quad1\quad\quad &$\wp$ & $\wp^2$\\
     \hline
    \hline
   1 & 0 & 0 & $1/2\omega^2$\\ \hline
    $\wp$ & 0 & $1/2\omega^2$  & 0\\ \hline
    $\wp^2$ & $1/2\omega^2$ & 0  & $g_2/8\omega^2$
  \end{tabular}
\end{center}

\begin{lemma}[Dubrovin \cite{Dub2dTFT}]  \label{th:flat}
Flat coordinates for the metric are given by
$$ t_1=c-\frac{\eta}{\omega},\quad t_2=\frac{1}{\omega},\quad t_3=\frac{\omega'}{\omega}.$$
\end{lemma}
\begin{proof}
This is (5.95) in \cite{Dub2dTFT}.  We simply use change of coordinates given by \eqref{basis} and the metric calculated above.  We have $\partial_c=\partial_{t_1}$.  From the identity
$$(\omega\partial_\omega+\omega'\partial_{\omega'})\frac{\eta}{\omega}=-2\frac{\eta}{\omega}$$
which uses the fact that $\omega\partial_\omega+\omega'\partial_{\omega'}$ is the degree operator, $\frac{\eta}{\omega}$ is homogeneous of degree -2 we have $\omega\partial_\omega+\omega'\partial_{\omega'}=\frac{2\eta}{\omega}\partial_{t_1}-\frac{1}{\omega}\partial_{t_2}$.  The identity
$$(\eta\partial_\omega+\eta'\partial_{\omega'})\frac{\eta}{\omega}=-\frac{1}{12}g_2-\frac{\eta^2}{\omega^2}$$
appearing as (C.69) in  \cite{Dub2dTFT} gives $
\eta\partial_\omega+\eta'\partial_{\omega'}=\left(\frac{1}{12}g_2+\frac{\eta^2}{\omega^2}\right)\partial_{t_1}-\frac{\eta}{\omega^2}\partial_{t_2}-\frac{i\pi}{2\omega^2}\partial_{t_3}$.
Hence we have
\beq\label{basis2}
\partial_{t_1}\mapsto 1,\quad -\frac{\eta}{\omega}\partial_{t_1}+\frac{1}{2\omega}\partial_{t_2}\mapsto\wp,\quad \left(\tfrac{1}{8}g_2-\frac{\eta^2}{2\omega^2}\right)\partial_{t_1}+\frac{\eta}{2\omega^2}\partial_{t_2}+\frac{i\pi}{4\omega^2}\partial_{t_3}\mapsto\wp^2.
\eeq

Hence the metric is given by:
\begin{center}
  \begin{tabular}{c || c | c | c }
      & \quad\quad$\partial_{t_1}$\quad\quad & $\partial_{t_2}$ & $\partial_{t_3}$\\
     \hline
    \hline
   $\partial_{t_1}$& 0 & 0 & $2/i\pi$\\ \hline
   $\partial_{t_2}$ & 0 & 2  & 0\\ \hline
    $\partial_{t_3}$&$2/i\pi$  & 0  & 0
  \end{tabular}
\end{center}
which is constant so that $\{t_1,t_2,t_3\}$ are flat coordinates.
\end{proof}
\begin{remark}
As mentioned in Section~\ref{sec:bergman} we can choose a different $(0,2)$ term $B(z,z')$ on the spectral curve \eqref{specweier} which still satisfies the compatibility condition \eqref{dosstest} by varying $b\in\C$.  For each $b$ it gives rise to a CohFT with the same genus 0 three-point function since ancestor invariants are coefficients of $B$-dependent differentials.  When $b$ is chosen so that $B(z,z')$ is normalised along a choice of cycle, e.g. $b=\eta'/\omega'$ so $\int_BB(z,z')=0$, then the CohFT is homogeneous and hence the same CohFT as for $b=\eta/\omega$.  Other choices of $b$ gives rise to non-homogeneous CohFTs.
\end{remark}

\begin{proposition}  \label{th:auxflat}
The flat coordinates correspond to the following auxiliary differentials:
\begin{align*}
dt_1&\longleftrightarrow T^1_0=(\omega\wp+b) dz-2\omega d\left(\frac{\wp^2}{\wp'}\right)+2\eta d\left(\frac{\wp}{\wp'}\right)+\left(\frac{\omega g_2}{4}+\frac{\eta^2}{\omega}\right)d\left(\frac{1}{\wp'}\right) \\
dt_2&\longleftrightarrow T^2_0=-d\left(\frac{\wp}{\wp'}\right)-\frac{\eta}{\omega}d\left(\frac{1}{\wp'}\right)\\
dt_3&\longleftrightarrow T^3_0=-\frac{i\pi}{2\omega}d\left(\frac{1}{\wp'}\right).
\end{align*}
\end{proposition}
\begin{proof}
The auxiliary differentials on the spectral curve corresponding to the normalized canonical basis are straighforward.  They are given by $V^i_kdz$ where
$$V^i_0=\frac{\wp(z-\omega_i)}{\sqrt{\wp''(\omega_i)}}$$
and for $k>0$, $V^i_k$ is the principal part of the $k$th derivative of $V^i_0$ with respect to $\wp(z)$.  We also have the canonical basis $U^i_0=\omega\wp(z-\omega_i)$.  The auxiliary differentials $T^i_kdz$ corresponding to flat coordinates are linear combinations of $V^i_kdz$
$$ V^i_k=\Psi^i_{\mu}\cdot T^{\mu}_k
$$
where we recall that $\Psi^i_{\mu}$ is the transition matrix from flat coordinates labeled by $\mu$ to normalized canonical coordinates labeled by $i$.  
We can calculate $\Psi$ via
$$\left(\begin{array}{ccc}1&1&1\\ \wp(\omega_1)&\wp(\omega_2)&\wp(\omega_3)\\\wp(\omega_1)^2&\wp(\omega_2)^2&\wp(\omega_3)^2\end{array}\right)\left(\begin{array}{c}\partial_{u_1}\\\partial_{u_2}\\\partial_{u_3}\end{array}\right)=\left(\begin{array}{ccc}1&0&0\\ -\frac{\eta}{\omega}&\frac{1}{2\omega}&0\\
\frac{1}{8}g_2-\frac{\eta^2}{2\omega^2}&\frac{\eta}{2\omega^2}&\frac{i\pi}{4\omega^2}\end{array}\right)\left(\begin{array}{c}\partial_{t_1}\\\partial_{t_2}\\\partial_{t_3}\end{array}\right)
$$
which we write as $M\frac{\partial}{\partial u}=T\frac{\partial}{\partial t}$ hence $M^{-1}T=\Delta^{1/2}\Psi^T$. 
The auxiliary differentials corresponding to $1,\wp,\wp^2$ in the Landau-Ginzburg model are:
\begin{align*}[U^1,U^2,U^3]\cdot M^{-1}=&\\
2\omega dz\left[\frac{\wp(z-\omega_1)+b}{\wp''(\omega_1)},\right.&\left.\frac{\wp(z-\omega_2)+b}{\wp''(\omega_2)},\frac{\wp(z-\omega_3)+b}{\wp''(\omega_3)}\right]
\hspace{-1mm}\left(\hspace{-1mm}\begin{array}{ccc}\wp(\omega_1)^2-\frac{1}{4}g_2&\wp(\omega_1)&1\\ \wp(\omega_2)^2-\frac{1}{4}g_2&\wp(\omega_2)&1\\ \wp(\omega_3)^2-\frac{1}{4}g_2&\wp(\omega_3)&1\end{array}\hspace{-1mm}\right)\\
=\left[-2\omega d\left(\frac{\wp^2}{\wp'}\right)\right.&\left.+(\omega\wp +b)dz+\frac{\omega g_2}{2}d\left(\frac{1}{\wp'}\right),-2\omega d\left(\frac{\wp}{\wp'}\right),-2\omega d\left(\frac{1}{\wp'}\right)\right]
\end{align*}
which is proven using the elliptic identities
$$\frac{\wp(z)^k}{\wp'(z)^2}=\sum_{i=1}^3\frac{\wp(\omega_i)^k\wp(z-\omega_i)}{\wp''(\omega_i)^2},\quad k=0,1,2$$
and slight generalisations for $k>2$.
Hence
\begin{align*}
[T^1,T^2,T^3]&=[U^1,U^2,U^3]\cdot M^{-1}\cdot T\\
&=\Big[-2\omega d\left(\frac{\wp^2}{\wp'}\right)+(\omega\wp+b) dz+\left(\frac{\omega g_2}{4}+\frac{\eta^2}{\omega}\right)d\left(\frac{1}{\wp'}\right) +2\eta d\left(\frac{\wp}{\wp'}\right),\\
&\hspace{5cm}-d\left(\frac{\wp+\eta/\omega}{\wp'}\right),
 -\frac{i\pi}{2\omega}d\left(\frac{1}{\wp'}\right)\Big]
\end{align*}
\end{proof}

The following lemma allows us to apply equation \eqref{eq:CohFT-SpectralCurve-match} to obtain ancestor invariants for the CohFT.
\begin{lemma}
The following kernels $K^i_0$ 
$$K^1_0=y(z),\quad K^2_0=-2\omega\zeta(z)+2\eta,\quad K^3_0=\frac{4}{i\pi}\left(\omega z\wp(z)^2-\left(\frac{\eta^2}{2\omega}+\frac{\omega}{8}g_2\right)z+\eta\zeta(z)\right)$$
are dual (as linear functionals) to $T^i_0$ for $i=1,2,3$, i.e.
$$\sum_{j=1}^3\res_{z=\omega_j}K^j_0(z)T^i_k(z)=\delta_{ij}\delta_{k0}.$$
\end{lemma}
\begin{proof}
Each kernel is analytic at $z=\omega_i$, $i=1,2,3$ and hence annihilates differentials analytic at $z=\omega_i$.  Consider the action of each kernel on $d(\wp^k/\wp')$ for $k=0,1,2$.
$$\sum_{j=1}^3\res_{z=\omega_j}K^i_0(z)d\left(\frac{\wp^k}{\wp'}\right)=-\sum_{j=1}^3\res_{z=\omega_j}dK^i_0(z)\frac{\wp^k}{\wp'}=\res_{z=0}dK^i_0(z)\frac{\wp^k}{\wp'}$$
so $K^1_0=y(z)=z/\omega$ annihilates $d(\wp^k/\wp')$ for $k=0,1$ and sends $d(\wp^2/\wp')$ to $-1/2\omega$.  Similarly $K^2_0=\zeta(z)$ annihilates $d(\wp^k/\wp')$ for $k=0,2$ and sends $d(\wp/\wp')$ to $1/2$.  Apply the kernels to $T^i_0$ given in Proposition~\ref{th:auxflat} as linear combinations of $d(\wp^k/\wp')$ (and terms analytic at $z=\omega_i$) to achieve the result.

The kernels $K^1_0$ and $K^2_0$ annihilate exact differentials that vanish to order 2 at $z=0$, in particular $T^i_k$ for $k>0$ by integration by parts.  One can also check that $K^3_0$ annihilates $T^i_k$ for $k>0$.
\end{proof}
{\em Remark.}  One can also produce kernels $K^i_j$ dual to each $T^i_j$.


The 3-point function in flat coordinates leads to the prepotential given in \cite{Dub2dTFT} (C.87):
$$F_0=\frac{1}{i\pi}t_1^2t_3+t_1t_2^2-\frac{i\pi}{2}t_2^4\left(\frac{1}{24}-\sum_{n=1}^{\infty}\frac{nq^n}{1-q^n}\right),\quad q=e^{2\pi i t_3}.
$$

\begin{proposition}
$$\exp F_1 =t_2^{1/8}\eta(q)^{1/4},\quad \eta(q)=q^{1/24}\prod_{n=1}^{\infty}(1-q^n).$$
\end{proposition}
\begin{proof}
Topological recursion---defined in Section~\ref{sec:TR}---applied to the spectral curve \eqref{specweier} uses the kernel 
\begin{align*}K(z_1,z)&=\frac{\omega}{2}\frac{\int^{z}_{\sigma_i(z)}(\wp(z_1-w)+b)dwdz_1}{(z-\sigma_i(z))\wp'(z)dz}\\&=\frac{\omega}{4}\frac{(\zeta(z_1+z)-\zeta(z_1-z)+2\eta_i+2b(z-\omega_i))dz_1}{(z-\omega_i)\wp'(z)dz}
\end{align*}
where $\sigma_i(z)=2\omega_i-z$.  Hence
\begin{align*}
\omega_{1,1}(z_1)&=\sum_{j=1}^3\res_{z=\omega_j}K(z_1,z)\wp(2z)=\sum_{j=1}^3\res_{z=\omega_j}\frac{\omega}{4}\frac{(\zeta(z_1+z)-\zeta(z_1-z))}{(z-\omega_i)\wp'(z)}\wp(2z)dzdz_1\\
&=\frac{\omega}{8}\left(2\sum_{j=1}^3\frac{\wp(\omega_i)\wp(z_0-\omega_i)}{\wp''(\omega_i)}-
\sum_{j=1}^3\frac{\wp(z_0-\omega_i)^2}{\wp''(\omega_i)}\right)dz_1\\
&=\frac{\omega}{8}\left(\frac{2\wp\wp''}{(\wp')^2}-\frac{(\wp'')^3}{(\wp')^4}+10g_2\frac{\wp}{(\wp')^2}+15g_3\frac{1}{(\wp')^2}+11\right)dz_1
\end{align*}
where $\eta_i\in\C$  and $b$ are annihilated by the residues.
Integrate the kernels $K^i_j$ against $\omega_{1,1}$ to get
$$\omega_{1,1}=0\cdot T^1_0+\frac{\omega}{8}T^2_0+\frac{i\eta\omega}{4\pi}T^3_0+\frac{1}{8}T^1_1+\frac{\eta}{4}T^2_1+\frac{g_2\omega^2-12\eta^2}{48i\pi}T^3_1.
$$
The primary part uses only $T^k_0$ and yields
$$F_1=\frac{1}{8}\log t_2 +f(t_3),\quad f'(t_3)=\frac{i\pi}{2}\left(\frac{1}{24}-\sum_{n=1}^{\infty}\frac{nq^n}{1-q^n}\right)$$
which is obtained from $\omega_{1,1}$ since $\frac{\partial F_1}{\partial t_2}=\frac{1}{8t_2}=\frac{\omega}{8}$ agrees with the coefficient of $T^2_0$ and $\frac{\partial F_1}{\partial t_3}=\frac{i\pi}{2}\left(\frac{1}{24}-\sum_{n=1}^{\infty}\frac{nq^n}{1-q^n}\right)=\frac{i\eta\omega}{4\pi}$ agrees with the coefficient of $T^2_0$.

\end{proof}

\section{General theory} \label{sec:general-theory}

In the preceding sections, we have investigated the construction of a global spectral curve producing the ancestor potential of a Frobenius manifold by topological recursion in some examples or assuming some additional properties of the curve defined by Dubrovin's superpotential. In the present section, we begin with the data of a semi-simple Frobenius manifold, and produce a global curve in a general setup not coming from the superpotential but rather from a family of curves built out of the reflection  group generated by the monodromies of the solutions of our Fuchsian system. In particular, it explains how our setup is related to the spectral curve built by Milanov in~\cite{MilanovGlobal}.

\subsection{Spectral curves from reflection group}
Here we define a family of spectral curves associated to the reflection group defined by the monodromies of the Fuchsian system given by Equation~\eqref{eq:Gauss-Manin-system-5-31-32}. The spectral curve defined by Dubrovin's superpotential is a particular point in this family. 

\begin{definition}
	For any $\gamma = (\gamma_1,\dots,\gamma_n) \in \mathbb{C}^n$, let us define a function $\phi^{[\gamma]} : \mathbb{C} \backslash \{L_i\} \to \mathbb{C}$ by
	\beq
	\phi^{[\gamma]}(\lambda;u) :=\sum_{i=1}^\mu \gamma_i \phi^{(i)}(\lambda;u)
	\eeq
	where $\phi^{(i)}$ are solutions to Equation~\ref{eq:Gauss-Manin-system-5-31-32} defined as in section~\ref{sec:DubrovinsSuperpotential}.
	
	We define the corresponding function $p^{[\gamma]}$  analytic  on $\mathbb{C}  \backslash \{L_i\} $ by
	\beq
	p^{[\gamma]}(\lambda,u) := \frac{\sqrt 2}{1-d} \left(\phi^{[\gamma]}\right)^T (U-\lambda) \Psi \un .
	\eeq
	
	Finally, let us define the pairing
	\beq
	\forall (\gamma,\gamma') \in \mathbb{C}^{2n} \, , \; (\gamma|\gamma') := - 2  \sum_{i,j} \gamma_i G^{ij} \gamma_j'.
	\eeq
	
\end{definition}

The main property of these functions is that $\phi^{[\gamma]}$ has the local behavior 
\begin{align} \label{eq:phi-locally-eq1-gen}
& \phi_j^{[\gamma]} = \frac{\sum_{i=1}^n \gamma_i G^{ij}}{\sqrt{u_j-\lambda}} + O(1)\ \text{for}\ \lambda\to u_j, & j=1,\dots,n;\\
\label{eq:phi-locally-eq2-gen}
& \phi_a^{[\gamma]} = \sum_{i=1}^n \gamma_i G^{ij} \sqrt{u_j-\lambda}\cdot O(1)\ \text{for}\ \lambda\to u_j, & a\not=j; a,j=1,\dots,n 
\end{align}
and $p^{[\gamma]}$ has  a local behavior for $\lambda\to u_i$ given by 
\begin{equation}\label{eq:p-lambda-locally-gen}
p^{[\gamma]}(\lambda,u) = p^{[\gamma]}(u_i,u)+ \sum_{j=1}^n \gamma_j G^{ji} \Psi_{i,\un} \sqrt{2(u_i-\lambda)}+O(u_i-\lambda),\quad i=1,\dots,n.
\end{equation}

Let $e_1,\dots,e_n$ be the standard basis of $\mathbb{C}^n$. We have $\phi^{[e_i]}(\lambda;u) = \phi^{(i)}(\lambda;u)$.

\begin{remark}
	Dubrovin's standard superpotential defined in Section~\ref{sec:DubrovinsSuperpotential} is obtained by considering the particular case $\gamma_j = \sum_{i=1}^n G_{ij}$.
\end{remark}

From now on, we assume that the reflections $\refl_i$ generate a finite group $W$.  Infinite families of Frobenius manifolds with finite group $W$ are given in \cite{Dub2dTFT}.

For any $\gamma \in \mathbb{C}^n$, one can define  a Riemann surface ${\cal D}^{[\gamma]}$ as a cover $\lambda^{[\gamma]}: {\cal D}^{[\gamma]} \to \mathbb{C}$ where $\lambda^{[\gamma]}(\tilde{p}^{[\gamma]},u)$ is the inverse function to $\tilde{p}^{[\gamma]}$ defined out of ${p}^{[\gamma]}$ by resolving the zeroes of $dp^{[\gamma]}$ as in Section~\ref{sec:DubrovinsSuperpotential}.  It is important to remark that the construction of $D^{[\gamma]}$ as a branched cover of $\mathbb{C} \backslash \{L_i\}$ does not depend on $\gamma$ but rather on a choice of gluing for the different sheets---see Remark~\ref{gluegenus} for a discussion of these choices. In this section, we consider the most naive gluing and the resulting spectral curve.
 
We consider the reflection $\refl_i$ as a linear map on the space $\mathbb{C}^n$ changing the coordinates of the vectors by the following rule:
\beq
\gamma_j \to \left\{ \begin{array}{l}
	\gamma_j \quad \hbox{if} \quad j\neq i \cr
	\gamma_i  + (\gamma | e_i) \quad \hbox{if} \quad j=i \cr
\end{array}
\right. .
\eeq
We denote $w\gamma$ the image of a vector $\gamma$ under the action of an element $w \in W$.

 We build the spectral curve ${\cal D}^{[\gamma]}$ as follows. A point $z \in {\cal D}^{[\gamma]}$ is defined by a pair $(\lambda,p)\in \hat{D} \times \mathbb{C}$ such that $p^{[\gamma]}(\lambda,u) = p$. By definition of $p^{[\gamma]}(\lambda,u)$, this defines a cover of $\hat{D}$ with ramification points in the fibres above the critical values $u_1,\dots,u_n$. We now glue in the most naive way, meaning that each point in the fibre above any of the $u_i$ is a simple ramification point. Let us now describe this sheeted cover.

Our spectral curve is obtained by analytic continuation of $p^{[\gamma]}$ from $\hat{D}$ through the (pre-images of the) cuts $L_i$. Each copy of $\hat{D}$ is then viewed as a sheet of a branched cover of the $\lambda$ plane. 
We can analytically continue $p^{[\gamma]}$ through $L_i$ seen as a cut on a Riemann surface giving rise to a new function of $\lambda$
\beq
p^{[\refl_i\gamma]}(\lambda):= \refl_i p^{[\gamma]}(\lambda;u) := \frac{\sqrt 2}{1-d} \left(\refl_i \phi^{[\gamma]}\right)^T (U-\lambda) \Psi \un
\eeq
where
\beq
\refl_i \phi^{[\gamma]}(\lambda;u)  = \sum_{j=1}^\mu \gamma_j \refl_{i} \phi^{(j)}(\lambda;u)
= \phi^{[\gamma]}(\lambda;u) + (\gamma|e_i) \phi^{[e_i]}(\lambda;u).
\eeq 
In other terms, we glue along the images of the cut  $L_i$ the sheets given by $p^{[\refl_i \delta]}(\lambda)$ and $p^{[\delta]}(\lambda)$  for all $\delta$ in the $W$-orbit of the initial vector $\gamma$.

The above procedure defines a $|W|$ sheeted cover ${\cal D}$ of the $\lambda$ plane such that the fiber above a point $\lambda$ is $\{p^{[w \gamma]}(\lambda,u)\}_{w \in W}$. The different sheets of this cover can thus be labelled by elements $w \in W$ and we denote by $\lambda^{[w]}$ the unique point in the fiber above a generic point $\lambda$ belonging to the sheet labelled by $w$. We define by $p$ the unique function on ${\cal D}$ such that
\beq
\forall w \in W \, , \; p\left(\lambda^{[w]}\right) = p^{[w\gamma]}(\lambda,u) 
\eeq
for a generic $\lambda$.

This cover is branched over all the points in the fibres above the points $u_i$, $i=1,\dots,n$, and a ramification point above $u_i$ joins the sheets labelled by $w$ and $\refl_i w$ for some $w \in W$.

This branched cover is our spectral curve. It has  $|W|/2$ simple ramification points over $u_i$, $i=1,\dots,n$. We denote by $u_{i}^{[w]}$, $w\in W$, the point in the fiber above $u_i$ such that $p\left(u_i^{[w]}\right) = p^{[w \gamma]}(u_i)$. This notation is ambiguous, so we denote by $W_i$ the minimal set such that 
\beq
\lambda^{-1}(u_i) = \{u_i^{w}\, | \,w \in W_i\}.
\eeq
By definition, one has the important relation
\beq
\forall w \in W_i \, , \; p\left(\lambda^{[w]}\right) - p\left(\lambda^{[\refl_i w]}\right) = 
{(w \gamma | e_i) \over (\gamma | e_i)} \left[ p\left(\lambda^{[Id]}\right) - p\left(\lambda^{[\refl_i]}\right) \right] .
\eeq

Thanks to our assumption of finiteness of $W$, ${\cal D}$ can be compactified by introducing ramification points of higher order above $\infty$.

The order of these ramification points above $\infty$ deserves some investigation. Since the reflection group $W$ is finite, then the ramification index of such a point is equal to the Coxeter number $h(W)$, i.e. the order of a Coxeter transformation. In such a case, there exists a longest positive root $\sum_i m_i \alpha_i$ (reminding that the set $\{\alpha_i\}$ is a set of simple roots) and the Coxeter number is equal to $1+\sum_i m_i$.

Let us recall as well that a Coxeter transformation is a product of all simple reflections. The different order for this product leading to different transformations, all with the same order. The different Coxeter transformations correspond to the different points in the fiber above $\infty$. In the case of an infinite group, this order is infinite and the different ramification points in the fiber above $\infty$ correspond to different conjugacy classes of Coxeter transformations.

We now have a Riemann surface $\Sigma$ which is a branched cover of the $\lambda$ plane. In our case, when the group is finite, its genus is given by the Riemann-Hurwitz formula:
\beq
2-2g(\Sigma) = 2 |W| - {|W| \over 2} n - (h(W)-1) {|W| \over h(W)}.
\eeq

\begin{remark}  \label{gluegenus}
	We have built {\bf a} curve using this procedure. There exist two ways of changing the cover built in this way. First by specifying some particular value for the vector $\gamma$. Second, by choosing a different gluing procedure for building the cover: for each point in the fibre above a critical value $u_i$, one can decide whether it is a ramification point or not. We followed here the most naive procedure where all the points are ramification points, recovering the spectral curve built by Milanov in the case of simple singularities~\cite{MilanovGlobal}. This procedure is the most general but gives the highest possible genus of the curve.
	
	In the preceding sections, we had chosen a particular value of $\gamma$ prescribed by Dubrovin's construction as well as the simplest possible curve by considering covers where only one point in the fibre above each critical value is a ramification point. This leads to the lowest genus spectral curve possible but requires one to study the gluing procedure carefully case by case.
\end{remark}

\subsection{Global topological recursion and correlation functions of a CohFT}

\subsubsection{Global topological recursion}

We remark that we are not in the cases discussed in the preceding sections since the spectral curve has $|W|/2$ ramification points in the fibre above one critical value. This implies that the topological recursion has to be modified a little in order to take the right form. 

\bd
We define the correlation functions defined by the global topological recursion applied to ${\cal D}$ as the differential forms defined by induction through
\begin{align*}
& \om_{g,k}(z_1,\dots,z_k) = \\
& \sum_{i=1}^n \sum_{w \in W_i}  \Res_{z \to u_{i}^{[w]}} {\int_{z}^{\sigma_{i,w}(z)} B(z_1,\cdot) \over 2 \left(\om_{0,1}(z) - \om_{0,1}(\sigma_i(z))\right)} \left[
\om_{g-1,k+1}(z,\sigma_{i,w}(z),z_2,\dots,z_k)
\phantom{ \sum_{h=1}^g}\right. \\
& \left. + \sum_{A\sqcup B = \{2,\dots,k\}} \sum_{h=0}^g \om_{h,|A|+1}(z, \vec z_A) \om_{g-h,|B|+1}(\sigma_{i,w}(z), \vec z_B) \right] ,
\end{align*}
where
\beq
\om_{0,1}(z) := p(z) d\lambda(z) ,
\eeq
\beq
\om_{0,2}(z_1,z_2) = \sum_{w \in W} (\gamma | w \gamma) B(z_1,z_2)
\eeq
and $\sigma_{i,w}$ is the local involution exchanging the two sheets meeting at $u_{i}^{w}$. In the right hand side, all the contributions involving a factor of $\om_{0,1}$ are set to $0$.

\ed

Note that, in this recursion, the recursion kernel does not involve $\om_{0,2}$ itself but rather $B$. This might seem to break the usual symmetry between the different arguments of $\om_{g,k}$ but, as we shall see in the next section, it is not the case.

\subsubsection{From global to local}

In \cite{DOSS12}, the correspondence between topological recursion and CohFT was discussed only at the local level. In order to match the correlation functions defined by the global topological recursion with those of the CohFT, let us translate the global recursion into a local one written in terms of integrals in the $\lambda$-plane around the critical values $u_i$.

\bl
The global topological recursion on the spectral curve ${\cal D}$ with $x = \lambda$, $y=p$ and $B(p_1,p_2)$ is equivalent to the local recursion with local spectral curve
\beq
\forall i=1, \dots n \, ,  \; \om_{0,1}^{[i]}(\lambda) = \Delta_{i,\lambda} p(\lambda^{[Id]}) d\lambda
\eeq
and
\beq
\forall i,j=1,\dots n \, , \; \om_{0,2}^{[i,j]}(\lambda_1,\lambda_2) = \Delta_{i,\lambda_1} \Delta_{j,\lambda_2} \om_{0,2}(\lambda_1^{[Id]},\lambda_2^{[Id]})
\eeq
where
\beq
\Delta_{i,\lambda} f(\lambda^{[w]}) = { f(\lambda^{[w]}) - f(\lambda^{[\refl_i w]}) \over 2}
\eeq
for a meromorphic form $f$ on ${\cal D}$.

In other words, the discontinuities
\beq
\om_{g,k}^{[i_1,\dots,i_k]}(\lambda_1,\dots,\lambda_k) :=  \prod_{j=1}^k \Delta_{i_j,\lambda_{j}} \om_{g,k}(\lambda_1^{[Id]},\dots,\lambda_k^{[Id]})
\eeq
of the correlation functions $\om_{g,k}$ produced by the global recursion satisfy the corresponding local recursion.
\el

\begin{proof}

	It is first important to note that
	\beq\label{disc-om-gn}
	\Delta_{i,\lambda_1} \om_{g,k+1}(\lambda^{[w]},\lambda_1^{[w_1]},\dots,\lambda_k) = {(w \gamma | e_i) \over (\gamma | e_i)} \Delta_{i,\lambda_1} \om_{g,k+1}(\lambda^{[Id]},\lambda_1^{[w_1]},\dots,\lambda_k).
	\eeq
	This is proved by induction and follows from the definition of $\om_{0,2}$ in terms of the Bergman kernel.
	This property allows us to rewrite the topological recursion in a local version where one sums only over one of the ramification points in the fiber above each of the critical values $u_i$.
	
	Writing $z = \lambda^{[w]}$, one can rewrite the term $\Res_{z \to u_{i}^{[w]}}$ as a residue when $\lambda \to u_i$ in the following way:
	\beq
	\Res_{\lambda^{[w]} \to u_i^{[w]}} = 2 \Res_{\lambda \to u_i}.
	\eeq
	This gives 
	\begin{align*}
	& \om_{g,k}(z_1,\dots,z_k) = \\
	& \left.\sum_{i=1}^n \sum_{w\in W_i}  \Res_{\lambda \to u_{i}} {\int_{\lambda^{[w]}}^{\lambda^{[\refl_i w]}} B(z_1,\cdot) \over  2 \Delta_{i,\lambda} \om_{0,1}(\lambda^{[w]})} 
	\Delta_{i,\lambda} \Delta_{i,\lambda'}\right[
	\om_{g-1,k+1}(\lambda^{[w]},\lambda'^{[w]},z_2,\dots,z_k) \\ 
	& \left.\left. 
	\sum_{A\sqcup B = \{2,\dots,k\}}  \sum_{h=0}^g  \om_{h,|A|+1}(\lambda^{[w]}, \vec z_A) \om_{g-h,|B|+1}(\lambda'^{[w]}, \vec z_B) \right] \right|_{\lambda' = \lambda}.
	\end{align*}
	Plugging property \eqref{disc-om-gn} into this equation, the global recursion reads
	\begin{align*}
	& \om_{g,k}(z_1,\dots,z_k) = \\
	& \left. \sum_{i=1}^n   \Res_{\lambda \to u_{i}} {{\sum\limits_{w\in W_i}} {(w \gamma | e_i) \over (\gamma|e_i)} \int_{\lambda^{[w]}}^{\lambda^{[\refl_i w]}} B(z_1,\cdot)  \over 2 \Delta_{i,\lambda} \om_{0,1}(\lambda^{[Id]})} 
	\Delta_{i,\lambda} \Delta_{i,\lambda'}\right[
	\om_{g-1,k+1}(\lambda^{[Id]},\lambda'^{[Id]},z_2,\dots,z_k)
	\\
	&\left.  \left.  \sum_{A\sqcup B = \{2,\dots,k\}} \sum_{h=0}^g \om_{h,|A|+1}(\lambda^{[Id]}, \vec z_A) \om_{g-h,|B|+1}(\lambda'^{[Id]}, \vec z_B) \right]\right|_{\lambda' = \lambda} .
	\end{align*}
	Finally, using the fact that $ 2 \sum_{w \in W_i} = \sum_{w \in W}$ in the expression above, one gets
	\begin{align*}
	& \om_{g,k}(z_1,\dots,z_k) = \\ &  
	\left. \frac 14  \sum_{i=1}^n   \Res_{\lambda \to u_{i}} { \Delta_{i,\lambda} \int^{\lambda^{[Id]}} \om_{0,2}(z_1,.)  \over \Delta_{i,\lambda} \om_{0,1}(\lambda^{[Id]})} 
	\Delta_{i,\lambda} \Delta_{i,\lambda'}\right[ 
	\om_{g-1,k+1}(\lambda^{[Id]},\lambda'^{[Id]},z_2,\dots,z_k)
	\\ & \left.\left.
	\sum_{A\sqcup B = \{2,\dots,k\}} \sum_{h=0}^g \om_{h,|A|+1}(\lambda^{[Id]}, \vec z_A) \om_{g-h,|B|+1}(\lambda'^{[Id]}, \vec z_B)
	\right]\right|_{\lambda' = \lambda} .
	\end{align*}
	Acting with the operators $\prod_{j=1}^k \Delta_{i_j,\lambda_j}$ on both sides proves the lemma.
\end{proof}

\subsubsection{Identification of the local initial data with a CohFT}

Now that we have derived a local topological recursion equivalent to the global one, one only needs to identify its initial data with the data of a CohFT following the dictionary of \cite{DOSS12}. For this purpose, we will follow exactly the same steps as in the preceding sections. Let us first state precisely the identification that we want to find since it is slightly different from the usual setup where one has only one ramification point in each fiber and a specific value for $\gamma$.

First of all, let us remind that, according to \cite{EynardIntersection},  the Laplace transform of the local two point function reads
\beq
\frac{1}{ 2 \pi \sqrt{\zeta_1 \zeta_2}} \iint\limits_{\substack{\lambda_1-u_i \in \mathbb{R} \\ \lambda_2-u_j \in \mathbb{R}}} \om_{0,2}^{[ij]}(\lambda_1,\lambda_2) e^{{\lambda_1 - u_i \over \zeta_1} + {\lambda_2 - u_j \over \zeta_2}}
= { \sum_{k=1}^n f(\zeta_1)_k^i \, f(\zeta_2)_k^j \over \zeta_1 + \zeta_2},
\eeq
where
\beq
f(\zeta)_k^i:= - {1 \over \sqrt{2\pi \zeta} } \int\limits_{\lambda_1-u_i \in \mathbb{R}}\left. {\om_{0,2}^{[ij]}(\lambda_1,\lambda_2) \over d\sqrt{-2 \lambda_2 +2u_j} }\right|_{\lambda_2 = u_k} e^{{\lambda_1 - u_i \over \zeta}}.
\eeq
In these terms, the identification consists in showing that
\beq\label{general-om02}
f(\zeta)_k^i = \sum_{j=1}^n (\gamma | e_j) G_{j i} R(\zeta)_k^i
\eeq
and
\beq\label{general-om01}
\sum_{j=1}^n (\gamma | e_j) G_{j i}  \sum_{k=1}^n R(\zeta)_k^i \Delta_{k}^{-{1\over 2}} = {1 \over \sqrt{2 \pi \zeta}} \int\limits_{\lambda - u_i \in \mathbb{R}} \om_{0,1}^{[i]}(\lambda) \cdot e^{\lambda - u_i \over \zeta} 
\eeq
where  $R(\zeta)$ is the $R$-matrix defining the CohFT we started from for deriving our Fuchsian system.

Note that the proof of Equation~\eqref{general-om01} is  a simple verbatim of the proof of Section~\ref{sec:genus0} by replacing $\phi^{(i)}$ by $\sum_{j=1}^n (\gamma|e_j) G_{ji} \phi^{(i)}$. A corollary of this identification is the following theorem:

\bt \label{thm:general-set-up}
The correlation functions $\om_{g,k}$ produced by the global recursion generate the correlation functions of the original CohFT through
\begin{align*}
& \om_{g,k}(\lambda(z_1)^{[Id]},\dots,\lambda(z_k)^{[Id]})= \\ & \sum_{\substack{i_1,\dots,i_k \\ j_1,\dots,j_k \\ d_1,\dots,d_k}} \prod_{l=1}^k \left[ (\gamma | e_j) G_{j i} \right]\int_{\overline{\mathcal{M}}_{g,k}} \alpha_{g,k}(e_{i_1},\dots,e_{i_k}) \prod_{l=1}^k
\psi_l^{d_l} d\left(\left(\frac{d}{dx}\right)^{d_l} \xi_{i_l}(z_l)\right).
\end{align*}
\et

\subsubsection{Compatibility condition and homogeneity}

Let us now remark the compatibility between Equation~\eqref{general-om02} and Equation~\eqref{general-om01} can be written as the usual compatibility condition for the Bergman kernel by considering all the ramification points, i.e.
\beq
\eta(z) = \sum_{i=1}^n \sum_{w \in W_i} \Res_{z'  = u_i^{[w]}} {dp \over d\lambda}(z')  B(z,z') +
\Res_{z'  = z} {dp \over d\lambda}(z')  B(z,z')  
\eeq
is invariant any local involution $\sqrt{\lambda(z)-u_i} \to - \sqrt{\lambda(z)-u_i}$. 

Finally, it is an easy exercise to prove the homogeneity at the level of $\om_{0,2}$ by using Rauch's variational formula as in Section~\ref{sec:highergenus}.

\appendix
\numberwithin{equation}{section}
\section{Frobenius manifolds of rank 2}
In this appendix we explicitly construct global spectral curves for two rank 2 CohFTs.  We begin with the prepotential $F(t_1,t_2)$ which one uses to produce the structure of a Frobenius manifold.  We follow Dubrovin's construction to produce a superpotential.  In both cases we need to vary the construction slightly due to degeneracy of the Gauss-Manin system.  The two examples satisfy the conditions of Theorem~\ref{thm:any-genus-identification} and hence topological recursion produces the CohFT associated to the Frobenius manifold.  Note that although the two examples are of genus zero, they do not satisfy the conditions of Theorem~\ref{thm:genus-0-identification}.

\subsection{Gromov-Witten invariants of $\C\mathbb{P}^1$}
$$
F=\dfrac{t_1^2\, t_2}{2}+e^{t_2},\quad E=t_1\partial_{t_1}+2\partial_{t_2},\quad E\cdot F=2F\ (+t_1^2)
$$
$$\eta^{\alpha\beta}=\left(\begin{array}{cc}0&1\\1&0\end{array}\right),\quad 
\Psi=\frac{1}{\sqrt{2}}\left(\begin{array}{cc}e^{-t_2/4}&e^{t_2/4}\\-ie^{-t_2/4}&ie^{t_2/4}\end{array}\right)
$$
$$\mu=\left(\begin{array}{cc}-\frac{1}{2}&0\\0&\frac{1}{2}\end{array}\right),\quad V=\Psi\mu\Psi^{-1}=\frac{i}{2}\left(\begin{array}{cc}0&-1\\1&0\end{array}\right)
$$
$$u_1=t_1+2e^{t_2/2},\quad u_2=t_1-2e^{t_2/2}$$
$$
V_1=\partial_{u_1}\Psi\cdot\Psi^{-1}=\frac{1}{u_1-u_2}\cdot V=-V_2
$$
The vector fields $\phi$ given in canonical coordinates satisfy \eqref{eq:Gauss-Manin-system-5-31-32} which is equivalent to the Fuchsian system:
\begin{equation}  \label{fuchsian}
(U-\lambda)\partial_{\lambda}\phi=(\tfrac{1}{2}+V)\phi
\end{equation}
and
\begin{equation}  \label{fuchscan}
\partial_{u_i}\phi=\left(-\frac{B_i}{\lambda-u_i}+V_i\right)\phi,\quad B_i=-E_i(\tfrac{1}{2}+V),\quad V_i=\partial_{u_i}\Psi\cdot\Psi^{-1}.
\end{equation}
This has general solution
$$\phi=\frac{c_1}{(u_1-u_2)^{1/2}}\left(\begin{array}{r}\sqrt{\dfrac{u_2-\lambda}{u_1-\lambda}}\\\hspace{-2mm}-i\sqrt{\dfrac{u_1-\lambda}{u_2-\lambda}}\end{array}\right)+\frac{c_2}{(u_1-u_2)^{1/2}}\left(\begin{array}{c}i\\1\end{array}\right).$$
We choose the solution $c_1=1$, $c_2=0$.  Since $d=1$ \eqref{eq:Superpotential-Formula}
does not apply.  Nevertheless, $\phi$ is the gradient of $p$ so we can calculate
$$dp(\lambda,u)=\frac{1}{u_1-u_2}\left(\sqrt{\dfrac{u_2-\lambda}{u_1-\lambda}}du_1+\sqrt{\dfrac{u_1-\lambda}{u_2-\lambda}}du_2\right).
$$

In this example, we will also go through the equivalent treatment in terms of flat coordinates for the pencil of metrics.  The vector fields $\phi$ are gradient vector fields of the flat coordinates 
$$\phi_i
=\Psi_{i\alpha}\eta^{\alpha\beta}\partial_{\beta}x(t_1-\lambda,t_2,...,t_n)$$
for the pencil of metrics $g-\lambda\eta$ where
$$g^{\alpha\beta}=\left(\begin{array}{cc}2e^{t_2}&t_1\\t_1&2\end{array}\right).$$
The flat coordinates for the pencil are of the form $x(t_1-\lambda,t_2,...,t_n)$ so it is enough to consider the case $\lambda=0$, i.e. find flat coordinates for the intersection form.  These are given by solutions of the Gauss-Manin system of linear differential equations ((5.9) in \cite{Dub98}):
$$ g^{\alpha\gamma}\partial_{\beta}\xi_{\gamma}+\sum_{\gamma}(\tfrac{1}{2}-\mu_{\gamma})c^{\alpha\gamma}\xi_{\gamma}=0,\quad \xi_{\beta}=\partial_{\beta}x.$$

$$\begin{array}{ccccccc}
2e^{t_2}\partial_1^2x&+&t_1\partial_1\partial_2x&+&0&=&0\\
2e^{t_2}\partial_1\partial_2x&+&t_1\partial_2^2x&+&e^{t_2}\partial_1x&=&0\\
t_1\partial_1^2x&+&2\partial_1\partial_2x&+&\partial_1x&=&0\\
t_1\partial_1\partial_2x&+&2\partial_2^2x&+&0&=&0
\end{array}
$$
$$\Rightarrow x=c_1\cdot\arccos\left(\tfrac{1}{2}t_1e^{-t_2/2}\right)+c_2\cdot t_2
$$
Choose
$$p=i\arccos\left(\tfrac{1}{2}(t_1-\lambda)e^{-t_2/2}\right)
$$
$$\lambda=t_1-2e^{t_2/2}\cos(-ip)=t_1-e^{t_2/2}(e^p+e^{-p})
$$
Note that the critical points of $\lambda$ are indeed $t_1\pm2e^{t_2/2}=u_{1/2}.$  It was proven in \cite{DOSS12} that the curve $p=\ln z$, $\lambda=a+b(z+1/z)$ does indeed produce the CohFT for Gromov-Witten invariants of $\C\mathbb{P}^1$.


\subsection{Discrete surfaces}\label{subsec:discretesurface}
The 2-dimensional Hurwitz-Frobenius manifold $H_{0,(1,1)}$ of double branched covers of the sphere, with two branch points and unramified at infinity was defined by Dubrovin \cite{Dub2dTFT}.  Its potential is 
$$
F=\dfrac{t_1^2\, t_2}{2}+\dfrac{1}{2}t_2^2\log t_2,\quad E=t_1\partial_{t_1}+2t_2\partial_{t_2},\quad E\cdot F=4F\ (+t_2^2)
$$
$$\eta^{\alpha\beta}=\left(\begin{array}{cc}0&1\\1&0\end{array}\right),\quad 
\Psi=\frac{1}{\sqrt{2}}\left(\begin{array}{cc}t_2^{1/4}&t_2^{-1/4}\\-it_2^{1/4}&it_2^{-1/4}\end{array}\right)
$$
$$\mu=\left(\begin{array}{cc}\frac{1}{2}&0\\0&-\frac{1}{2}\end{array}\right),\quad V=\Psi\mu\Psi^{-1}=\frac{i}{2}\left(\begin{array}{cc}0&-1\\1&0\end{array}\right)
$$
$$u_1=t_1+2t_2^{1/2},\quad u_2=t_1-2t_2^{1/2}$$
$$
V_1=\partial_{u_1}\Psi\cdot\Psi^{-1}=\frac{-1\ }{u_1-u_2}\cdot V=-V_2
$$
The general solution of \eqref{fuchsian} and \eqref{fuchscan} is
$$\phi=\frac{c_1}{(u_1-u_2)^{1/2}}\left(\begin{array}{r}\sqrt{\dfrac{u_2-\lambda}{u_1-\lambda}}\\\hspace{-2mm}i\sqrt{\dfrac{u_1-\lambda}{u_2-\lambda}}\end{array}\right)+\frac{c_2}{(u_1-u_2)^{1/2}}\left(\begin{array}{c}i\\1\end{array}\right).$$
The solutions of Dubrovin described in \eqref{eq:DubrovinProp-1}-\eqref{eq:DubrovinProp-4} yield $\phi^{(1)}=\phi^{(2)}$ hence $G^{ij}$ is degenerate.  We use one of the solutions $\phi=\phi^{(1)}$ in \eqref{eq:Superpotential-Formula} to get
$$p(\lambda,u)=\frac{t_2^{1/4}}{2} \frac{\sqrt{(u_1-\lambda)(u_2-\lambda)}}{(u_1-u_2)^{1/2}} \left(\begin{array}{cc}1&i\end{array}\right) \left(\begin{array}{c}1\\-i\end{array}\right)=\dfrac{1}{2}\sqrt{ (u_1-\lambda)(u_2-\lambda) }.
$$
This corresponds to the spectral curve $\lambda=t_1+z+t_2/z$, $p=z-t_2/z$ which arises from the well-studied Hermitian matrix model with Gaussian potential hence discrete maps \cite{EO09} and was shown to correspond to the given CohFT in \cite{ACNP}.



\end{document}